\documentclass[a4paper, 10pt]{article}

\usepackage{graphics,graphicx}
\usepackage{amssymb,amsmath}
\usepackage{fullpage}
\usepackage{enumitem}
\usepackage{xcolor}
\usepackage{stmaryrd}

\usepackage[capitalize]{cleveref}
 \usepackage{thmtools}

\newtheorem{theorem}{Theorem}

\newtheorem{lemma}{Lemma}

\newtheorem{problem}{Problem}

\newtheorem{claim}{}[lemma]

\newtheorem{subclaim}{}[claim]
\newtheorem{cor}{Corollary}

\usepackage{thmtools}

\def \no {\noindent}

\def \sm {\setminus}
\def \es {\emptyset}

\newenvironment{proof}[1][]%
{\noindent {\setcounter{equation}{0}\it Proof.
	}{#1}{}}{\hfill$\Box$\vspace{2ex}}

\def\longbox#1{\parbox{0.85\textwidth}{#1}}
\newcommand*\sq{\mathbin{\vcenter{\hbox{\rule{0.75ex}{1.0ex}}}}}

\begin{document}
\title{On near optimal colorable graphs}

\author{C.~U.~Angeliya\thanks{Computer Science Unit, Indian Statistical
Institute, Chennai Centre, Chennai 600029, India. This research is   supported by National Board of Higher Mathematics (NBHM), DAE, India.  } \and Arnab Char\thanks{Computer Science Unit, Indian Statistical
Institute, Chennai Centre, Chennai 600029, India. } \and T.~Karthick\thanks{Corresponding author, Computer Science Unit, Indian Statistical
Institute, Chennai Centre, Chennai 600029, India. Email: karthick@isichennai.res.in. ORCID: 0000-0002-5950-9093. This research is partially supported by National Board of Higher Mathematics (NBHM), DAE, India. }}


\maketitle
\begin{abstract}
A class of graphs $\cal G$ is said to be \emph{near optimal colorable}  if there exists a constant $c\in \mathbb{N}$ such that every graph $G\in \cal G$ satisfies $\chi(G) \leq \max\{c, \omega(G)\}$, where $\chi(G)$ and $\omega(G)$ respectively denote the chromatic number and clique number of $G$. The class of near optimal colorable graphs   is an important subclass of the class of $\chi$-bounded graphs which is well-studied in the literature. In this paper, we show that  the class of  ($F, K_4-e$)-free graphs is near optimal colorable, where $F\in \{P_1+2P_2,2P_1+P_3,3P_1+P_2\}$  and the graph $K_4-e$ is commonly referred as the {\em diamond}. This partially answers a question of Ju and Huang [Theoretical Computer Science 993 (2024) Article No.: 114465] and is related to a question of Schiermeyer (unpublished).  Furthermore, using these results with some earlier known results, we also provide an alternate proof to the fact that the \textsc{Chromatic Number} problem for
	the class of ($F, K_4-e$)-free graphs is solvable in  polynomial time, where $F\in \{P_1+2P_2,2P_1+P_3,3P_1+P_2\}$.
\end{abstract}

\medskip
\no{\bf Keywords}:~Graph classes; Vertex Coloring; Chromatic number; Clique number; $\chi$-boundedness; Forbidden induced subgraphs.

\section{Introduction}

We consider only finite graphs  with no loops or multiple edges. We refer to West \cite{west} for standard notation and terminology, and  we refer to \cite{BLS} for   undefined special graphs used in this paper.
For an integer $\ell\geq 1$, the graphs $P_{\ell}$ and $K_{\ell}$, respectively denote  the induced  path and the complete graph on $\ell$ vertices. For an integer $\ell\geq 3$, the graph $C_{\ell}$ denotes the induced  cycle  on $\ell$ vertices. An \emph{odd-hole} or an \emph{odd-hole of length $\ell$} is the graph $C_{\ell}$ with $\ell\geq 5$ and $\ell$ is odd. A $K_{\ell}-e$ is the graph obtained from $K_{\ell}$, $\ell\geq 2$ by removing an edge. In the literature, the graph $K_4-e$ is commonly referred as the {\em diamond}.  The
\emph{union} of two vertex-disjoint graphs $G_1$ and $G_2$, denoted by $G_1+G_2$, is the graph
with vertex-set $V(G_1)\cup V(G_2)$ and edge-set $E(G_1)\cup E(G_2)$.  The
union of $\ell$ copies of the same graph $G$ will be denoted by $\ell G$; for
instance, $2P_2$ denotes the graph that consists of two disjoint copies
of $P_2$.  Given a set of graphs, say  ${\cal F}= \{H_1, H_2, \ldots, H_{t}\}$, we say that a graph $G$ is \emph{$(H_1, H_2, \ldots, H_{t})$-free} if  no induced subgraph of $G$ is isomorphic to $H_i$, $i\in \{1,2, \ldots, t\}$;  if  ${\cal F} = \{H\}$, then we simply write    $G$  is
{\em $H$-free} instead of    $(H)$-free.

Given a graph $G$ and a positive integer $k$, a  {\it $k$-coloring}  of $G$ is a mapping $\psi: V(G) \rightarrow \{1, 2, \ldots, k\}$ such that $\psi(u)\neq \psi(v)$ if $uv\in E(G)$. The smallest integer $k$ for which $G$ admits a
$k$-coloring is called the \emph{chromatic number} of $G$, and is denoted by $\chi(G)$.  We say that a graph $G$ is  {\it $k$-colorable} if $\chi(G)\leq k$.    A \emph{clique} in a graph $G$ is a set of mutually adjacent vertices, and the \emph{clique number} of $G$ (denoted by $\omega(G)$), is the largest possible integer $t$ such that $G$ contains a clique
of size $t$. A graph $G$ is \emph{perfect} if
$\chi(H) = \omega(H)$ for every induced subgraph $H$ of $G$. Examples of perfect graphs include bipartite graphs, split graphs, chordal graphs, the class of $P_3$-free graphs, etc.

The \textsc{Coloring} problem asks whether  a graph admits a $k$-coloring or not for some given integer $k$. The \textsc{Chromatic Number} problem asks for the least possible integer $k$ for which a graph admits a $k$-coloring.
Given an arbitrary class of graphs, \textsc{Coloring} is well-known to be \textsf{NP}-complete for any  ﬁxed  $k\geq 3$. However, Gr\"otschel, Lov\'asz  and Schrijver \cite{GLS} showed that \textsc{Chromatic Number} is polynomial time solvable when restricted to perfect graphs.
Thus researchers are interested in studying \textsc{Coloring}  and \textsc{Chromatic Number}  for some    special classes of graphs, and we  refer to \cite{GJPS} for an excellent survey of recent results and open problems on computational complexity of coloring graphs.

A class of graphs is said to be {\em hereditary} if it is closed under taking induced subgraphs.
A (hereditary) class $\cal G$ of graphs is called
\emph{$\chi$-bounded} with \emph{$\chi$-binding function} $f:~\mathbb{N}\rightarrow \mathbb{N}$ (where $f(1)=1$ and $f(x)\geq x$, for all $x\in \mathbb{N}$) if $\chi(G) \leq f(\omega(G))$ holds whenever $G \in \cal  G$. If $\cal G$ is  a $\chi$-bounded class of graphs such that for each $x\in \mathbb{N}$, there exists a graph $G\in \cal G$ with $\omega(G)=x$, then the  \emph{smallest $\chi$-binding function} for $\cal{G}$ is the function $f^*:\mathbb{N} \rightarrow \mathbb{N}$ defined by $f^*(x) := \max\{\chi(G)\mid G\in {\cal G} \mbox{ and } \omega(G)=x\}$ for all $x\in \mathbb{N}$. Several classes of graphs are known be $\chi$-bounded; see \cite{AK-Survey, SR-Poly-Survey}.
An important class among classes of $\chi$-bounded  graphs is the following.
 A class of graphs $\cal G$ is said to be \emph{near optimal colorable} \cite{Ju-Huang} if there exists a constant $c\in \mathbb{N}$ such that every graph $G\in \cal G$ satisfies $\chi(G) \leq \max\{c, \omega(G)\}$.  For instance, the class of perfect graphs is clearly near optimal colorable.
  Also since every ($2P_2$, gem)-free graph $G$ has $\chi(G)\leq  \max\{3, \omega(G)\}$ \cite{BRSV}, and since every ($P_6$, paw)-free graph $G$ has $\chi(G)\leq \max\{4, \omega(G)\}$ \cite{Olariu,RST},  we see that the class of ($2P_2$, gem)-free graphs and  the class of  ($P_6$, paw)-free graphs are  near optimal colorable. Gy\'arf\'as \cite{Gyarfas} showed that every ($2P_1+P_2, K_4-e$)-free graph $G$ satisfies $\chi(G)\leq \max\{3,\omega(G)\}$, and hence the class of ($2P_1+P_2, K_4-e$)-free graphs is also near optimal colorable.
		For any two graphs $H_1$ and $H_2$, a characterization for the
	near optimal colorability of $(H_1, H_2)$-free graphs with three exceptional cases was given in \cite{Ju-Huang}, and we 
  are interested in the below problem.	%
	\begin{problem} [\cite{Ju-Huang}] \label{JH-Problem}
		Decide whether the class of $(F, K_t-e)$-free graphs is near optimal colorable when $F$ is a forest (which is not an induced subgraph of a $P_4$)   and $t\geq 4$.
	\end{problem}

Gei\ss er  \cite{Geiber} showed that every ($P_5, K_4-e$)-graph $G$ satisfies $\chi(G)\leq \max\{3,\omega(G)\}$. Recently, the second and third authors of this paper \cite{AK} generalized this result and showed that every ($P_5, K_5-e$)-free graph $G$ satisfies $\chi(G)\leq \max\{7,\omega(G)\}$.
Goedgebeur, Huang, Ju and
Merkel \cite{GHJM} proved that every ($P_6$, $K_4-e$)-free graph $G$ satisfies   $\chi(G) \leq \max\{6, \omega(G)\}$. Hence the class of ($P_5, K_4-e$)-free graphs, the class of ($P_5, K_5-e$)-free graphs, and the class of ($P_6$, $K_4-e$)-free graphs are near optimal colorable. Here we focus on  \cref{JH-Problem} when $F$ is a forest on $5$ vertices and $t=4$. Recently Ju and Huang \cite{Ju-Huang25} showed that if $F$ contains a $K_{1,3}$, then the class of ($F, K_4-e$)-free graphs are not near optimal colorable. Thus from   earlier mentioned  results,    \cref{JH-Problem} is open for
the class of ($F$, $K_4-e$)-free graphs, where $F\in \{P_1+2P_2,  2P_1+P_3,  3P_1+P_2,   5P_1\}$  and for the class of ($P_7, K_4-e$)-graphs.
Moreover, Schiermeyer (unpublished) asked the following.
\begin{problem}
 Does there exist an integer $\lambda\geq 3$ such that   every ($P_7, K_4-e$)-free graph $G$ satisfies
$\chi(G) \leq \omega(G) + \lambda$?
\end{problem}

\begin{figure}
\centering
 \includegraphics{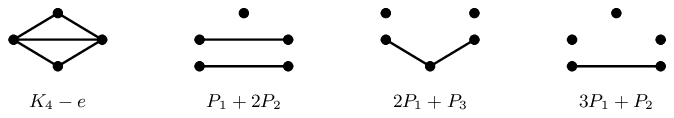}
\caption{Some special graphs.}\label{fig-basic}
\end{figure}

In this paper, we show that the class of ($F$, $K_4-e$)-free graphs  is near optimal colorable, where $F\in \{P_1+2P_2,2P_1+P_3,3P_1+P_2\}$ (see Figure~\ref{fig-basic}); so   the only open case for \cref{JH-Problem} when $|V(F)|\leq 5$ and $t=4$  is that $F\cong 5P_1$.
  Note the class of ($F, K_4-e$)-free graphs, where $F\in \{P_1+2P_2,2P_1+P_3\}$  is a subclass of the class of ($P_7, K_4-e$)-free graphs, and  includes the class of ($2P_1+P_2, K_4-e$)-free graphs \cite{Gyarfas}.
It is long-known that  if $F\in \{P_1+2P_2,2P_1+P_3,3P_1+P_2\}$, then  every ($F$, $K_3$)-free graph $G$ satisfies $\chi(G)\leq 3$; see  \cite{Rand-Thesis}. It is also known from a more general result of Chudnovsky et al. \cite{CKLV}  that  every ($3P_1+P_2, K_4-e$)-free graph $G$ with $\omega(G)\geq 9$ is $\omega(G)$-colorable (see also \cite{Ju-Huang25}). Here we prove the following (see Figure~\ref{fig-basic}).

\begin{theorem}\label{col-thm-0}
The function $f:\mathbb{N}\rightarrow \mathbb{N}$ defined by $$f(1)=1,~ f(2)=3, ~f(3)=4 \mbox{ and } f(x)=x, \mbox{ for }  x\geq 4,$$
	is the smallest $\chi$-binding function for the class of ($P_1+2P_2,K_4-e$)-free graphs.
\end{theorem}

\begin{theorem}\label{col-thm-2}
Let $G$ be a  ($2P_1+P_3, K_4-e$)-free graph with $\omega(G)\geq 3$. Then the following hold:
 \begin{enumerate}[label=  (\roman*), leftmargin=1.25cm] \itemsep=0pt
 \item If $\omega(G)=3$, then $\chi(G)\leq 5$.  \item If $\omega(G)\geq 4$, then $\chi(G)\leq \max\{6, \omega(G)\}$.
 \end{enumerate} Moreover, the bound is tight for
$\omega(G)\geq 6$.
\end{theorem}

\begin{theorem}\label{col-thm}
Let $G$ be a  ($3P_1+P_2, K_4-e$)-free graph with $\omega(G)\geq 3$. Then the following hold:
 \begin{enumerate}[label=  (\roman*), leftmargin=1.25cm] \itemsep=0pt
 \item If $\omega(G)=3$, then $\chi(G)\leq 5$.  \item If $\omega(G)\geq 4$, then $\chi(G)\leq \max\{7, \omega(G)\}$.
 \end{enumerate} Moreover, the bound is tight for
$\omega(G)\geq 7$.
\end{theorem}

As corollaries of our results, we show that \textsc{Chromatic Number} for
	the class of ($F$, $K_4-e$)-free graphs can be solved in  polynomial time, where $F\in \{P_1+2P_2,2P_1+P_3,3P_1+P_2\}$.  We note that this fact may also be obtained from a result of Dabrowski, Dross and Paulusma \cite{DDP} which states that the class of ($F$, $K_4-e$)-free graphs has bounded clique-width, and from a result of Rao \cite{Rao} which states that \textsc{Chromatic Number} is polynomial time solvable for graphs of bounded clique-width. In this respect, the following may be noted, where $F\in \{P_1+2P_2,2P_1+P_3,3P_1+P_2\}$.  Clearly  both the class of ($2P_1+P_3$)-free graphs and the class of  ($3P_1+P_2$)-free graphs include  the class of $4P_1$-free graphs, and the class of ($P_1+2P_2$)-free graphs includes the class of ($2P_1+P_2$)-free graphs.

 \begin{itemize}[leftmargin=0.5cm]\itemsep=0pt
 \item \textsc{Chromatic Number} is  \textsf{NP}-hard   for the class of $F$-free graphs, and remains \textsf{NP}-hard even for an apparently small class of $F$-free graphs, namely the class of $(4P_1, 2P_2, 2P_1+P_2, P_1+K_3)$-free graphs; see \cite{Schn}.

 \item The computational complexity of \textsc{Chromatic Number} is unknown for the class of ($4P_1, K_{1,3}$)-free graphs, for the class of ($4P_1, K_{1,3}, 2P_1+P_2$)-free graphs, and for the class of ($4P_1, C_4$)-free graphs; see \cite{DDP,LM}.

      \item \textsc{Chromatic Number} is  \textsf{NP}-hard   for the class of ($K_4-e$)-free graphs, and remains   \textsf{NP}-hard for    ($K_4-e$, $H$)-free graphs when $H$ contains a cycle or a $K_{1,3}$; see \cite{KKTW}.

       \end{itemize}

We finish this section with some notation and terminology used in this paper.  For a positive integer $k$, we simply write $[k]$ to denote the set $\{1, 2, \ldots , k\}$. In a graph $G$, the \emph{neighborhood} of a vertex $v$ is the set
$N_G(v):=\{u\in V(G)\setminus \{v\}\mid uv\in E(G)\}$. We let $\overline{N}_G(v):= V(G)\setminus (N_G(v)\cup \{v\})$.
The degree of a vertex $v$ in $G$, denoted by $deg_G(v)$, is the number of vertices adjacent to
it, that is, the cardinality of $N_G(v)$.  We drop the
subscript $G$ when there is no ambiguity. The complement graph of a graph $G$ will be denoted by $\overline{G}$.
A
\emph{stable set} in a graph $G$ is a set of mutually  nonadjacent vertices in $G$.

Given any two graphs $G$ and $H$, we say that $G$ \emph{contains} $H$, if $H$ is an induced subgraph of $G$.
 For a graph $G$ and a set $S\subseteq V(G)$,  $G[S]$
denotes  the subgraph induced by $S$ in $G$, and $G-S$ denotes  the subgraph $G[V(G)\sm S]$.  Given a graph $H$, we say that a set {\em $S\subseteq V(G)$ induces   $H$}, if $G[S]$   is isomorphic to $H$.
For any two  disjoint subsets
$A$ and $B$ of $V(G)$, we denote by $[A,B]$, the set of edges which has
one end in $A$ and the other end in $B$.  The edge-set   $[A,B]$ is \emph{complete} or $A$ is \emph{complete} to $B$, if each vertex in $A$
is adjacent to every vertex in $B$. We say that $A$ is \emph{anticomplete} to $B$ if $[A, B]=\es$. If $A$ is a singleton set,
  say $\{v\}$, then we simply write $v$ is complete (anticomplete) to  $B$ instead of $\{v\}$ is  complete (anticomplete) to $B$.
The edge-set $[A, B]$ is \emph{special} if each vertex in $A$ is adjacent to at most one vertex in $B$, and vice-versa. We say that the pair \{$A$, $B$\} is \emph{graded} if  $A$ and $B$ are cliques  and $[A, B]$ is  special.


 \medskip
The remainder of the paper is organized as follows.   In \cref{sec:use}, we give some crucial and useful results, and in \cref{sec:P1+2P2} we prove \cref{col-thm-0} and its consequences. In \cref{genprop}, we present some  important structural  properties of ($K_4-e$)-free graphs that contain a
$C_5$ and use them in the latter sections. Finally in \cref{sec:2P1+P3,sec:3P1+P2}, we prove \cref{col-thm-2,col-thm} and their consequences by using related interesting  structure theorems (see \cref{main-thm-2,main-thm}).

\section{Some useful results}\label{sec:use}
 In this section, we give some results  which we use often in the latter sections.  Since every bipartite graph is a perfect graph,  a result of Lov\'asz \cite{Lovasz} implies the following.
 \begin{lemma}[\cite{Lovasz}]\label{comp-bi}
The complement graph of a bipartite graph is perfect.
 \end{lemma}

Indeed, a celebrated result of Chudnovsky et al.~\cite{spgt} gives a characterization for the class of perfect graphs in  terms of forbidden induced subgraphs, and is given below.

\begin{theorem}[\cite{spgt}]\label{spgt} (Strong Perfect Graph Theorem) A graph $G$ is perfect if and only if $G$ does not
contain an odd-hole or the complement graph of an odd-hole.
\end{theorem}

 We say that a graph $G$ is a \emph{good graph} if   $V(G)$ can be partitioned into three mutually disjoint cliques, say $Q_1$, $Q_2$ and $Q_3$,
			such that $\{Q_1,Q_2\}$, $\{Q_2,Q_3\}$ and $\{Q_3,Q_1\}$ are graded.
Good graphs play  a crucial role in this paper, and we will prove that good graphs with clique number $\omega\geq 4$ can be colored using $\omega$ colors.   Note that the size of a maximum stable set in a good graph $G$ is at most 3, and   the size of any clique containing  vertices both from  $Q_1$ and from $Q_2\cup Q_3$ (say) is at most $3$. We begin with the following lemma.

\begin{lemma}\label{4clqsizlem}
    If $G$ is a good graph with $\omega(G)=4$, then $G$ is $4$-colorable.
	
\end{lemma}
\begin{proof}
    Let $G$ be a good graph with $\omega(G)=4$. We use the same notation as in the definition. Recall that the size of a maximum stable set in $G$ is at most 3, and   the size of any clique containing  vertices both from  $Q_1$ and from $Q_2\cup Q_3$ (say) is at most $3$.  So since $\omega(G)=4$, we may assume that $|Q_1|=4$, $|Q_2|\leq 4$ and $|Q_3|\leq 4$.  Since $G[Q_1\cup Q_2]$ is the complement graph of a bipartite graph, $G[Q_1\cup Q_2]$ is perfect, by \cref{comp-bi}. Also since $Q_1$ is a clique with $|Q_1|=4$, clearly $\omega(G[Q_1\cup Q_2]) =4$. Hence $\chi(G[Q_1\cup Q_2]) = \omega(G[Q_1\cup Q_2]) =4$.  If $Q_3=\es$, then we are done; so we may assume that $Q_3\neq \es$. Likewise, we may assume that $Q_2\neq \es$.
    Next we claim the following:

   \begin{claim}\label{grad-1} If $|Q_3|\in \{1,2\}$, then $G$ is $4$-colorable. Likewise, if $|Q_2|\in \{1,2\}$, then $G$ is $4$-colorable.\end{claim}

 \no{\it Proof of  \cref{grad-1}}.~By symmetry, it is enough to prove the first statement. Suppose that $|Q_3|\in \{1,2\}$. If $|Q_3|=1$, then we let $Q_3:=\{q_3^1\}$, and if $|Q_3|=2$, then we let $Q_3:=\{q_3^1,q_3^2\}$. Since $\chi(G[Q_1\cup Q_2]) =4$, there are four nonempty disjoint stable sets, say $S_1$, $S_2$, $S_3$ and $S_4$ such that $S_1\cup S_2\cup S_3\cup S_4=Q_1\cup Q_2$.
  Now  since   $Q_i$'s are pairwise graded, $q_3^1$ has at most one neighbor in each of $Q_1$ and $Q_2$, and hence there exists an $i\in \{1,2,3, 4\}$  such that $q_3^1$ is anticomplete to $S_i$.  We may assume (without loss of generality) that $i=1$.  Also  since   $Q_i$'s are pairwise graded, $q_3^2$ has at most one neighbor in each of $Q_1$ and $Q_2$, and hence there exists a $j\in \{2,3, 4\}$ such that $q_3^2$ is anticomplete to $S_j$. We may assume (without loss of generality) that  $j=2$. Now $S_1\cup \{q_3^1\}$,  $S_2\cup \{q_3^2\}$, $S_3$ and $S_4$ are   stable sets whose union is $V(G)$.
  Thus   $G$ is $4$-colorable, and we are done.   This proves \cref{grad-1}.
 $\sq$

\medskip
By \cref{grad-1}, we may assume that $|Q_3|\in \{3,4\}$ and $|Q_2|\in \{3,4\}$. Next:

    \begin{claim}\label{grad-2}  If $|Q_3|=3$, then $G$ is $4$-colorable. Likewise, if $|Q_2|=3$, then $G$ is $4$-colorable.  \end{claim}

\no{\it Proof of  \cref{grad-2}}.~By symmetry, it is enough to prove the first statement. Suppose that $|Q_3|=3$. We let $Q_3: =\{q_3^1,q_3^2,q_3^3\}$. By \cref{grad-1}, we may assume that $Q_1 \cup Q_2 \cup\{q_3^1,q_3^2\}$ can be partitioned into four  nonempty disjoint  stable sets, say $R_1$, $R_2$, $R_3$ and $R_4$, such that $q_3^1\in R_1$ and $q_3^2\in R_2$. So $(R_1\sm \{q_3^1\})\cup (R_2\sm \{q_3^2\})\cup R_3\cup R_4\subseteq Q_1\cup Q_2$, and hence $|R_3|\leq 2$ and $|R_4|\leq 2$. If $q_3^3$ is anticomplete to $R_3$, then $R_3\cup \{q_3^3\}$ is a stable set, and hence $R_1$, $R_2$, $R_3\cup \{q_3^3\}$ and $R_4$ are stable sets whose union is $V(G)$, and hence $G$ is $4$-colorable; so we may assume that  $q_3^3$ has a neighbor in $R_3$, say $q^*$. Likewise, we may assume that $q_3^3$ has a neighbor in  $R_4$, say $q^{**}$.  Then since $Q_i$'s are pairwise graded, we may assume (without loss of generality) that $q^*\in Q_1$ and $q^{**}\in Q_2$. Again since $Q_i$'s are pairwise graded, clearly $\{q_3^1,q_3^2\}$ is anticomplete to $\{q^*,q^{**}\}$, and since $(R_1\sm \{q_3^1\})\cup (R_2\sm \{q_3^2\})\subseteq Q_1\cup Q_2$, we see that $q_3^3$ is anticomplete to $(R_1\sm \{q_3^1\})\cup (R_2\sm \{q_3^2\})$. Also since $Q_i$'s are pairwise graded and $|R_3\sm \{q^*\}|\leq 1$, there exists a $k\in \{1,2\}$   such that $q_3^k$ is anticomplete to $R_3$. We may assume (without loss of generality) that $k=1$.   Then  $(R_1\sm \{q_3^1\})\cup \{q_3^3\}$, $R_2$, $R_3\cup \{q_3^1\}$ and $R_4 $ are four stable sets whose union is $V(G)$, and  hence $G$ is $4$-colorable. This proves \cref{grad-2}. $\sq$

 \medskip
 By \cref{grad-1} and \cref{grad-2}, we may assume that $|Q_i|=4$ for $i\in \{1,2,3\}$.  We let $Q_3 =\{q_3^1,q_3^2,q_3^3,q_3^4\}$. By \cref{grad-2}, we may assume that $Q_1 \cup Q_2 \cup\{q_3^1,q_3^2,q_3^3\}$ can be partitioned into four  nonempty disjoint  stable sets, say $T_1$, $T_2$, $T_3$ and $T_4$  such that $q_3^1\in T_1$, $q_3^2\in T_2$ and $q_3^3\in T_3$. So $(T_1\sm \{q_3^1\})\cup (T_2\sm \{q_3^2\})\cup (T_3\sm \{q_3^3\})\cup T_4\subseteq Q_1\cup Q_2$, and hence   $|T_4|=2$. We let $T_4=\{q, q'\}$. If $q_3^4$ is anticomplete to $T_4$, then $T_4\cup \{q_3^4\}$ is a stable set, and hence $T_1$, $T_2$,  $T_3$ and $T_4\cup \{q_3^4\}$ are stable sets whose union is $V(G)$, and hence $G$ is $4$-colorable; so we may assume that  $q_3^4$ has a neighbor in $T_4$, say $q$.   Also if $q_3^4q'\in E(G)$, then since   $Q_i$'s are pairwise graded, clearly $q_3^1$ is anticomplete to $T_4$, and $q_3^4$ is anticomplete to $T_1\sm \{q_3^1\}$,  and hence $(T_1\sm \{q_3^1\})\cup \{q_3^4\}$, $T_2$, $T_3$
 and $T_4\cup \{q_3^1\}$ are stable sets whose union is $V(G)$, and hence $G$ is $4$-colorable; so we may assume that $q_3^4q'\notin E(G)$.  Then  since   $Q_i$'s are pairwise graded, we see that $Q_3\sm \{q_3^4\}$ is anticomplete to $q$, and there exists a $p\in \{1,2,3\}$  such that $q_3^pq'\notin E(G)$. We may assume (without loss of generality) that $p=1$.  Then $T_4\cup \{q_3^1\}= \{q,q', q_3^1\}$ is a stable set.   Now if $q_3^4$ is anticomplete to $T_1\sm \{q_3^1\}$,  then $(T_1\sm \{q_3^1\})\cup \{q_3^4\}$, $T_2$, $T_3$ and $T_4\cup \{q_3^1\}$ are four stable sets whose union is $V(G)$, and hence $G$ is $4$-colorable; so we may assume that $q_3^4$ has a neighbor in $T_1\sm \{q_3^1\}$, say $q^*$. So since $|T_1|\leq 3$, we have $|T_1\sm \{q_3^1, q^*\}|\leq 1$. Then since $Q_i$'s  are pairwise graded, we observe that $q_3^4$ is anticomplete to $(T_2\sm \{q_3^2\}) \cup (T_3\sm \{q_3^3\})$,  $\{q_3^2,q_3^3\}$ is anticomplete to $q^*$, and there exists an $\ell\in \{2,3\}$ such that $q_3^{\ell}$ is anticomplete to $T_1\sm \{q_3^1, q^*\}$. We may assume (without loss of generality) that $\ell=2$.    Then clearly $(T_1\sm \{q_3^1\})\cup \{q_3^2\}$, $(T_2\sm \{q_3^2\})\cup \{q_3^4\}$, $T_3$ and $T_4\cup \{q_3^1\}$ are four stable sets whose union is $V(G)$, and hence $G$ is $4$-colorable. This proves \cref{4clqsizlem}.
\end{proof}

\begin{figure}
\centering
 \includegraphics[height=3cm]{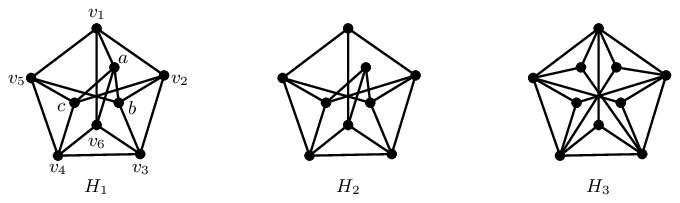}
\caption{Some graphs with $\chi=4$ and $\omega=3$.  }\label{fig-opt}
\end{figure}

	\begin{theorem}\label{lem-good}
    If $G$ is a good graph with $\omega(G)\geq 4$, then  $\chi(G)=\omega(G)$.
	
\end{theorem}

	\begin{proof}
    Let $G$ be a good graph with $\omega(G)\geq 4$.   We prove the theorem by induction on $\omega(G)$. If $\omega(G)=4$, then the theorem follows from \cref{4clqsizlem}. So we let $\omega(G)=k\geq 5$. Suppose that the theorem holds for all good induced subgraphs $H$  of $G$  with $4\leq \omega(H)\leq k-1$.  Let $S$ be a maximum stable set of $G$. Then clearly $S$ is a stable set such that $\omega(G-S) =\omega(G)-1$, and hence $\chi(G-S)=\omega(G-S)$, by induction hypothesis. Thus $\chi(G)= \chi(G-S)+\chi(G[S]) = (\omega(G)-1)+1 = \omega(G)$. 	This proves \cref{lem-good}.
\end{proof}

We remark that the assumption that the clique number $\omega\geq 4$ in \cref{lem-good} cannot be dropped, and that the constant 4 cannot be lowered. Indeed, there is a good graph $G$ with $\omega(G)=3$ such that $G$ is not $3$-colorable. For instance, take $G\cong H_1$ (see Figure~\ref{fig-opt}). Then $\chi(G)=4$, $\omega(G)=3$ and clearly $V(G)$ can be partitioned into three cliques, namely $Q_1:=\{v_1,v_6,a\}$, $Q_2:=\{v_2,v_3,b\}$ and $Q_3:=\{v_4,v_5,c\}$,
			such that $\{Q_1,Q_2\}$, $\{Q_2,Q_3\}$ and $\{Q_3,Q_1\}$ are graded.

\medskip
For our next lemma, we need the following definition.

\medskip
 \no{\bf Graph class $\cal C$}:~We say that  a   graph $G\in {\cal C}$ (see Figure~\ref{fig-classC})  if its vertex-set can be partitioned into a set consisting of five vertices, say $\{u_1,u_2,\ldots, u_5\}$, and six mutually disjoint sets, say $Q_1$, $Q_1'$, $Q_2$, $Q_2'$, $Q_3$ and $Q_3'$, such that we have the following:

 \vspace{-0.2cm}
\begin{itemize}\itemsep=0pt
 \item $u_1u_2,u_2u_3, u_3u_4, u_4u_5$ and $u_5u_1$   are edges in $G$,

\item  for each $\ell\in \{1,2,3\}$, $|Q_{\ell}'|\leq 1$ and $|Q_{\ell}\cup Q_{\ell}'|\geq 3$,
\item  $Q_1\cup Q_1'\cup \{u_1,u_2\}$, $Q_2\cup Q_2'\cup \{u_3,u_4\}$ and   $Q_3\cup Q_3'\cup \{u_1,u_5\}$ are cliques,

 \item $Q_1'$ is complete to $u_{4}$, $Q_2'$ is complete to $u_{1}$, and $Q_3'$ is complete to $u_{3}$,

 \item the pairs  $\{Q_1, Q_2\}$ and $\{Q_2, Q_3\}$ are graded, and
 \item  no other edges in $G$.
\end{itemize}

\begin{figure}[h]
\centering
 \includegraphics[height=5cm]{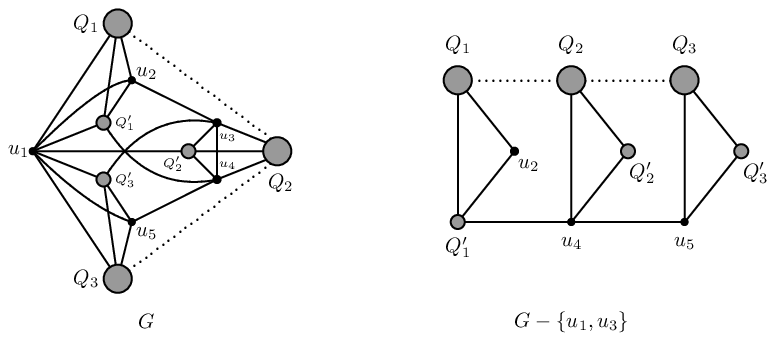}
\caption{Schematic representation of a graph $G$ in $\cal C$, and the graph obtained from $G$ by deleting the vertices $u_1$ and $u_3$. Here, a shaded circle represents a clique, a solid line between any two shapes represents that the corresponding vertex-sets   are complete to each other, a dotted line  between two shaded circles represents that the corresponding pair of cliques are graded,   and the absence of a line between   two shapes represents that the corresponding vertex-sets are anticomplete to each other.}\label{fig-classC}
\end{figure}

\begin{lemma}\label{classC-col}
For any graph $G\in {\cal C}$, we have $\chi(G)=\omega(G)$.
\end{lemma}
\begin{proof}
Let $G\in \cal C$. We use the same notation as in the above definition. Since $|Q_1\cup Q_1'|\geq 3$ and since $Q_1\cup Q_1'\cup \{u_1,u_2\}$ is a clique, we have $\omega(G)\geq 5$. Clearly by the definition of $G$ (see also Figure~\ref{fig-classC}),  the only possible maximum cliques of $G$ are $Q_1\cup Q_1'\cup \{u_1,u_2\}$, $Q_2\cup Q_2'\cup \{u_3,u_4\}$, and $Q_3\cup Q_3'\cup \{u_1,u_5\}$. Now we let $S:=\{u_1,u_3\}$. Then $S$ is a stable set in $G$ that intersects with all maximum cliques of $G$ exactly once, and so $\omega(G-S) = \omega(G)-1\geq 4$.
Also by the definition of $G$,   clearly the sets $Q_1\cup Q_1' \cup \{u_2\}$, $Q_2\cup Q_2'\cup \{u_4\}$  and $Q_3\cup Q_3'\cup \{u_5\}$ are pairwise graded  (see Figure~\ref{fig-classC}) whose union is $V(G)\sm S$. Thus $G-S$ is a good graph with $\omega(G-S)\geq 4$, and so by \cref{lem-good}, we have  $\chi(G-S)=\omega(G-S)=\omega(G)-1$. Then since $S$ is a stable set in $G$, we have $\chi(G)\leq \chi(G-S)+\chi(G[S])= (\omega(G)-1)+1 =\omega(G)$. This proves \cref{classC-col}.
\end{proof}

We will also use the following known result.

 \begin{theorem}[\cite{Ju-Huang}]\label{JH-thm}
 If $\cal G$ is a given  near optimal colorable hereditary class of graphs where
	every   $G \in \cal G$ satisfies   $\chi(G)\leq \max\{c, \omega(G)\}$, for some  $c\in \mathbb{N}$, and if \textsc{Coloring} for $\cal G$ is solvable in  polynomial time  for every fixed positive integer $k \leq c-1$, then  \textsc{Chromatic Number} for
	$\cal G$ can be solved in  polynomial time.
\end{theorem}

\section{The class of ($P_1+2P_2, K_4-e$)-free graphs}\label{sec:P1+2P2}
In this section, we prove \cref{col-thm-0} and its consequences. First we prove the following.

\begin{lemma} \label{2K2K1bound}
If $G$ is a  ($P_1+2P_2, K_4-e$)-free graph with $\omega(G)\geq 3$, then $\chi(G)\leq \max\{4, \omega(G)\}$.
\end{lemma}

\begin{proof}
 	Let $G$ be a $(P_1+2P_2,K_4-e)$-free graph.  Let $K$ be a maximum  clique in $G$, so that $|K|\geq 3$.  Let $C\subseteq K$ be such that $|C|=3$, say $C=\{v_1,v_2,v_3\}$. Let

\begin{center}
 \begin{tabular}{ccl}
 $R_{i}$ &:=  & $\{u\in V(G)\sm C \mid N(u)\cap C=\{v_i\}\}$, for $i\in [3]$,\\
 $M$  &:=  & $\{u\in V(G)\sm C \mid N(u)\cap C=C\}$, and \\
 $M'$&:=  & $\{u\in V(G)\sm C \mid N(u)\cap C=\es\}$.
 \end{tabular}
 \end{center}Let $R:=R_1\cup R_2\cup R_3$.  Throughout this lemma,   the arithmetic operations on indices are understood to be modulo $3$. Note that $K\sm C \subseteq M$. Moreover the following hold.

\begin{claim}
	\label{K2K1part} $V(G)=C\cup R\cup M\cup M'$.
\end{claim}
\no{\it Proof of \cref{K2K1part}}.~Suppose to the contrary that there is  a vertex  in $V(G)\sm (C\cup R\cup M\cup M')$, say $u$. Then $u$  has exactly two neighbors in $C=\{v_1,v_2,v_3\}$, and consequently $\{v_1,v_2,v_3,u\}$ induces a $K_4-e$ which is a contradiction. So \cref{K2K1part} holds.~$\sq$

\begin{claim}\label{K2K1Ri}
For each $i\in [3]$, the set $R_i$ is    a clique or a stable set.
\end{claim}
\no{\it Proof of \cref{K2K1Ri}}.~By symmetry, it is enough to prove the statement for $i=1$. Since $\{v_{2},v_{3}\}$ is anticomplete to $R_1$ and since $G$ is ($P_1+2P_2$)-free, we have $G[R_1]$ is ($P_1+P_2$)-free. Also since $v_1$ is complete to $R_1$ and since $G$ is ($K_4-e$)-free, we have $G[R_1]$ is $P_3$-free. Since $P_3$-free graphs can be expressed as the union of (disjoint) complete graphs, we conclude that if $R_1$ is   not a stable set, then it must be a clique. This proves \cref{K2K1Ri}. $\sq$

\begin{claim}\label{K2K1Mclq}
	The set $M$ is a clique such that $\omega(G[M])\leq \omega(G)-3$, and $R$ is anticomplete to $M$.
\end{claim}
\no{\it Proof of \cref{K2K1Mclq}}.~If there are nonadjacent vertices in $M$, say $m_1$ and $m_2$, then $\{m_1,v_2,v_3,m_2\}$ induces a $K_4-e$; so $M$ is a clique. Then since $M$ is complete to $C$, we have $\omega(G[M])\leq \omega(G)-3$. Next if there are adjacent vertices, say $r\in R$ and $m\in M$, then we may assume that $r\in R_1$ and then $\{r, m,v_1,v_2\}$ induces a $K_4-e$; so  $R$ is anticomplete to $M$.  This proves \cref{K2K1Mclq}.~$\sq$

\begin{claim}\label{K2K1Riclqnbd}
	For each $i\in [3]$  and for any   $u\in R_i\cup M'$,  $N(u)\cap R_{i+1}$ and $N(u)\cap R_{i+2}$ are stable sets.
\end{claim}
\no{\it Proof of \cref{K2K1Riclqnbd}}.~By symmetry, it is enough to prove the statement for $i=1$. So $u \in R_1\cup M'$. If there are adjacent  vertices in $N(u)\cap R_{2}$, say $r$ and $r'$, then $\{r,r',u,v_{2}\}$ induces a $K_4-e$; so $N(u)\cap R_{2}$ is a stable set. Likewise, $N(u)\cap R_{3}$ is also a stable set. This proves \cref{K2K1Riclqnbd}. $\sq$

\begin{claim}\label{K2K1per}
	We may assume that $R\cup M'$ is not a clique.
\end{claim}
\no{\it Proof of \cref{K2K1per}}.~Suppose that $R\cup M'$ is a clique. Then from \cref{K2K1part} and \cref{K2K1Mclq}, $V(G)$ can be partitioned into two cliques, namely $C\cup M$ and $R\cup M'$. Hence $G$ is the complement graph of a bipartite graph and thus $\chi(G)=\omega(G)$ (by \cref{comp-bi}), and we are done. So we may assume that $R\cup M'$ is not a clique.~$\sq$

\begin{claim}\label{K2K1Rimaxclqsz}
  For each $i\in [3]$, we have  $\omega(G[R_{i}\cup R_{i+1}])\leq 2$. Moreover for each $i\in [3]$, $\chi(G[R_{i}\cup R_{i+1}])\leq 2$, and for any $m\in M$, we have $\chi(G[R_{i}\cup R_{i+1}\cup \{m,v_{i+2}\}])\leq 2$.
\end{claim}
 \no{\it Proof of \cref{K2K1Rimaxclqsz}}.~We begin by proving the first assertion. By symmetry, it is enough to prove it for $i=1$. Suppose to the contrary  that $\omega(G[R_{1}\cup R_{2}])\geq 3$. Then by using \cref{K2K1Riclqnbd}, we may assume (without loss of generality) that  there are three mutually adjacent vertices in $R_{2}$, say $r_2,r_2'$ and $r_2''$.    Then from \cref{K2K1Ri}, $R_2$ is a clique.  Thus since $\{r_2, r_2', r_2'',v_2\}$ is a clique, we have $\omega(G)\geq 4$, and so $K\sm C\neq \es$. Let $u\in K\sm C$. If there is a vertex in $R_3$, say $r_3$, then by \cref{K2K1Riclqnbd}, we may assume that $r_2r_3,r_2'r_3\notin E(G)$, and then $\{r_3,r_2,r_2',u,v_1\}$ induces a $P_1+2P_2$ (by \cref{K2K1Mclq}); so we have $R_3=\es$. Likewise, we have $R_1=\es$. Since $R\cup M' = R_2\cup M'$, we have from  \cref{K2K1per} that $R_2\cup M'$ is not a clique and therefore   $M'\neq \es$. Let $m\in M'$. Then from \cref{K2K1Riclqnbd}, we may assume that $mr_2,mr_2'\notin E(G)$. Then $\{m,r_2,r_2',v_1,v_3\}$ induces a $P_1+2P_2$ which is a contradiction. So $\omega(G[R_{1}\cup R_{2}])\leq 2$. This proves  the first assertion.

We now prove the second assertion. Again by symmetry, it is enough to prove it for $i=1$. Since $R_{1}\cup R_{2}$ is anticomplete to $M\cup \{v_3\}$ (by \cref{K2K1Mclq}), it is enough to prove that $\chi(G[R_{1}\cup R_{2}])\leq 2$. If $R_{1}$ and $R_{2}$ are stable sets, then we are done. Also, if both $R_{1}$ and $R_{2}$ are not stable sets, then from \cref{K2K1Ri},  \cref{K2K1Riclqnbd} and from the first assertion, it is easy to see that $G[R_{1}\cup R_{2}]$ is a subgraph of a $C_4$, and hence $\chi(G[R_{1}\cup R_{2}])\leq 2$.
So   we may assume that $R_{1}$ is   a stable set and that $R_{2}$ is not a stable set. By \cref{K2K1Ri}, $R_{2}$ is a clique. Since $\omega(G[R_{2}])\leq 2$, we let $R_{2}:=\{r_2,r_2'\}$.
Now from \cref{K2K1Riclqnbd}, $R_{1}\cap N(r_2)$ is anticomplete to $r_2'$.  Then  $(R_{1}\cap N(r_2))\cup \{r_2'\}$ and $(R_{1}\sm N(r_2))\cup \{r_2\}$ are stable sets, and hence the second assertion holds.~$\sq$

\begin{claim}\label{K2K1RiMcomcol}
We have $\chi(G[R_{\ell}\cup M'])\leq 2$, for some $\ell\in [3]$.
\end{claim}
\no{\it Proof of \cref{K2K1RiMcomcol}}.~Suppose to the contrary that for each $\ell\in [3]$, we have $\chi(G[R_{\ell}\cup M'])\geq 3$.  Since for $\ell\in [3]$, $R_{\ell}\cup M'$ is anticomplete to $\{v_{\ell+1},v_{\ell+2}\}$, $G[R_{\ell}\cup M']$ is $(P_1+P_2)$-free and hence $G[R_{\ell}\cup M']$ is a complete multipartite graph. Thus since $\chi(G[R_{\ell}\cup M'])\geq 3$, for each $\ell\in [3]$, $G[R_{\ell}\cup M']$ contains a $K_3$, say with vertex-set $\{a_{\ell},b_{\ell},c_{\ell}\}$. By \cref{K2K1Riclqnbd}, we may assume that $a_{\ell},b_{\ell}\in M'$. Next we  show that for each  $\ell\in [3]$, $R_{\ell}\cup M'$ is a clique. By symmetry, it is enough to prove this for $\ell=1$. Suppose to the contrary that there are nonadjacent vertices in $R_1\cup M'$, say  $u$ and $v$. Since  $G[R_1\cup M']$ is a complete multipartite graph, there are three stable sets which are mutually complete to each other, say $S_1, S_2$ and $S_3$ such that $a_1\in S_1$, $b_1\in S_2$ and $c_1\in S_3$. Then we see that there is a $j\in [3]$  such that $\{u,v\}$ is complete to $S_{j}\cup S_{j+1}$, and then $S_{j}\cup S_{j+1}\cup \{u,v\}$ induces a $K_4-e$ which is  a contradiction. So $R_{\ell}\cup M'$ is a clique for each $\ell\in [3]$. Hence from \cref{K2K1per}, we may assume that there are nonadjacent vertices, say $r_1\in R_1$ and $r_2\in R_2$, and then $\{a_1,b_1,r_1,r_2\}$ induces a $K_4-e$ which is a contradiction. So  there is an  $\ell\in [3]$  such that $\chi(G[R_{\ell}\cup M'])\leq 2$.~$\sq$

\medskip
From \cref{K2K1RiMcomcol}, we may assume that $\chi(G[R_1\cup M'])\leq 2$. Hence $\chi(G[R_1\cup M'\cup \{v_2,v_3\}])\leq 2$.
Now if $\omega(G)=3$, then $M=\es$ and so from  \cref{K2K1Rimaxclqsz}, we have $\chi(G)\leq \chi(G[R_1\cup M'\cup \{v_2,v_3\}])+\chi(G[R_2\cup R_3\cup \{v_1\}])\leq 2+2=4$, and we are done. So we may assume that $\omega(G)\geq 4$. Then $K\sm C\neq \es$, and let $u\in K\sm C$. Then from  \cref{K2K1Mclq} and \cref{K2K1Rimaxclqsz}, we have $\chi(G)\leq \chi(G[R_1\cup M'\cup \{v_2,v_3\}])+\chi(G[R_2\cup R_3\cup \{u,v_1\}])+\chi(G[M\sm \{u\}]) \leq 2+2+(\omega(G)-4)=\omega(G)$. This proves \cref{2K2K1bound}.  \end{proof}

 \noindent{\bf Proof of \cref{col-thm-0}}.~By \cref{2K2K1bound} and the fact that every $(P_1+2P_2,K_3)$-free graph $G$ satisfies $\chi(G)\leq 3$ \cite{Rand-Thesis},  clearly $f$ is a $\chi$-binding function for the class of ($P_1+2P_2,K_4-e$)-free graphs.
	 Next we prove that $f$ is the smallest $\chi$-binding function for the class of ($P_1+2P_2,K_4-e$)-free graphs.  First note that for each $t\in \mathbb{N}$, the graph $K_t$  is a  ($P_1+2P_2,K_4-e$)-free graph with $\omega(K_t)=t$. Now we let $G_1\cong K_1$, $G_2\cong C_5$, $G_3\cong H_1$ or $G_3\cong H_2$  (see  Figure~\ref{fig-opt}), and for $\ell\geq 4$, we let $G_\ell\cong K_\ell$. Then  clearly for each $\ell\in \mathbb{N}$, $G_\ell$ is  a ($P_1+2P_2,K_4-e$)-free graph such that $\omega(G_\ell)=\ell$ and $\chi(G_\ell) =f(\ell)$.   This proves \cref{col-thm-0}. \hfill{$\Box$}

\medskip
From \cref{col-thm-0}, we immediately have the following.

 \begin{cor}\label{P12P2-NOC}
The class of ($P_1+2P_2, K_4-e$)-free graphs is near optimal colorable. That is, every ($P_1+2P_, K_4-e$)-free graph $G$ satisfies $\chi(G)\leq \max\{4, \omega(G)\}$.
\end{cor}

Bonomo et al. \cite{BCMSSZ} showed that \textsc{Coloring} for  the class of ($P_1+2P_2$, $K_4-e$)-free graphs is solvable in  polynomial time   for every fixed positive integer $k \leq 3$. This result together with  \cref{P12P2-NOC} and \cref{JH-thm}  imply  the following.

 \begin{cor}
  \textsc{Chromatic Number} for
	the class of ($P_1+2P_2$, $K_4-e$)-free graphs can be solved in  polynomial time.
 \end{cor}

\section{Properties of  ($K_4-e$)-free graphs  that  contain a  $C_5$}\label{genprop}

In  this section, we prove some general structural properties  of a ($K_4-e$)-free graph  that  contains a $C_5$, and use them in the latter sections.
Throughout this section, the   arithmetic operations on indices are understood to be modulo $5$.

Let $G$ be a ($K_4-e$)-free graph.
Suppose that $G$ contains a  $C_5$, say with
 vertex-set $C:= \{v_1, v_2, v_3, v_4,$ $ v_5\}$ and edge-set $\{v_1 v_2, v_2 v_3, v_3 v_4, v_4 v_5, $ $v_5 v_1\}$. Then since $G$ is ($K_4-e$)-free, we immediately see that, for all $i\in [5]$, each vertex in $V(G)\setminus C$   is nonadjacent to at least one of $v_{i}$, $v_{i+1}$ and $v_{i+2}$. Using this fact, we     define the following sets: For $i \in [5]$, we let:
\begin{center}
 \begin{tabular}{ccl}
 $A_{i}$ &:=  & $\{ v \in V(G)\setminus C \mid N(v) \cap  C  = \{v_i\}\},$\\
 $B_{i}$ &:=  & $\{v \in V(G)\setminus C \mid  N(v) \cap C = \{v_i, v_{i+1}\}\},$\\
 $D_{i}$ &:=  & $\{ v \in V(G)\setminus C \mid N(v) \cap C = \{v_{i-1}, v_{i+1}\}\},$\\
 $Z_{i}$ &:= & $\{ v \in V(G)\setminus C \mid N(v) \cap C = \{v_i, v_{i+1},v_{i-2}\}\},$ \\
 $X_i$ &:=  & $B_i\cup Z_{i}$, and\\
 $T $ &:=  & $\{ v \in V(G)\setminus C \mid N(v) \cap C = \es\}$.\\
 \end{tabular}
 \end{center}
 Further we let $A:= \cup_{i=1}^5A_i$,   $B:= \cup_{i=1}^5B_i$, $D:= \cup_{i=1}^5D_i$, $Z:= \cup_{i=1}^5Z_i$ and $X:= \cup_{i=1}^5X_i$. So $V(G)= C\cup A\cup B\cup D\cup Z\cup T=C\cup A\cup D\cup X\cup T$. Moreover note that for each $i\in [5]$, $N(v_i) = \{v_{i-1},v_{i+1}\}\cup A_i\cup B_{i-1}\cup B_i \cup D_{i-1}\cup D_{i+1}\cup Z_{i-1}\cup Z_{i}\cup Z_{i+2} = \{v_{i-1},v_{i+1}\}\cup A_i  \cup D_{i-1}\cup D_{i+1}\cup X_{i-1}\cup X_{i}\cup Z_{i+2}$.  Also for each $i\in [5]$, the following hold:

\begin{enumerate}[label=  $(\mathbb{O}\arabic*)$, leftmargin=1.25cm] \itemsep=0pt
\item \label{DZ} {\it $D_{i}$ is a stable set, and $|Z_{i}|\leq 1$.}
 \item \label{X-matching}  {\it $\{ X_i \cup \{v_i,v_{i+1}\}, X_{i+2}\cup\{v_{i-2},v_{i+2}\}\}$ is graded. Likewise, $\{X_i \cup \{v_i,v_{i+1}\}, X_{i-2}\cup \{v_{i-2},v_{i-1}\}\}$ is graded.}
\item  \label{bto}
 {\it $[X_{i}, A_{i}\cup A_{i+1}\cup B_{i-1}\cup B_{i+1}\cup (D\sm D_{i-2})\cup (Z\sm Z_{i})]$ is an empty set.}
\end{enumerate}

        \begin{proof}
      \ref{DZ}:~If there are  adjacent vertices  in $D_i$, say $d$ and $d'$, then $\{v_{i-1}, d, d', v_{i+1}\}$ induces a $K_4-e$; so $D_i$ is a stable set.
Next suppose to the contrary that there are two vertices in $Z_i$, say $z$ and $z'$. Now if  $z$ and $z'$ are adjacent, then $\{v_{i}, z, z', v_{i-2}\}$ induces a  $K_4-e$, and if $z$ and $z'$   are nonadjacent, then $\{z, v_{i}, v_{i+1},z'\}$ induces a $K_4-e$.  These contradictions show that $|Z_i| \leq 1$. So \ref{DZ} holds.  $\sq$

\smallskip
\noindent{\ref{X-matching}:}~If there are nonadjacent vertices in $X_i$, say $x$ and $x'$, then $\{x, v_i, v_{i+1}, x'\}$ induces a $K_4-e$; so $X_i$ is a clique. So by the definition of $X_i$, clearly $X_i \cup \{v_i,v_{i+1}\}$ is a clique.  Likewise, $X_{i+2}\cup \{v_{i-2},v_{i+2}\}$ is also a clique.  So again by the definition of $X_i$'s and \ref{DZ}, it is enough to show that $[X_i, X_{i+2}]$ is special. Now  if there is a vertex in $X_i$, say $x_i$, which has two neighbors in $X_{i+2}$, say $x_{i+2}$ and $x_{i+2}'$, then $\{x_i,x_{i+2},x_{i+2}',v_{i+2}\}$ induces a $K_4-e$; so each vertex in $X_i$ is adjacent to at most one vertex of $X_{i+2}$.  Likewise, each vertex in $X_{i+2}$ is adjacent to at most one vertex of $X_{i}$.  So $\{ X_i \cup \{v_i,v_{i+1}\}, X_{i+2}\cup\{v_{i-2},v_{i+2}\}\}$ is graded. Likewise, $\{X_i \cup \{v_i,v_{i+1}\}, X_{i-2}\cup \{v_{i-2},v_{i-1}\}\}$ is graded. This proves \ref{X-matching}.  $\sq$

\smallskip
\noindent{\ref{bto}:}~If  there are   vertices, say $x\in X_{i}$ and $u\in A_{i}\cup A_{i+1}\cup B_{i-1}\cup B_{i+1}\cup (D\sm D_{i-2})\cup (Z\sm Z_{i})$ such that $xu\in E(G)$, then $\{v_{i},v_{i+1},  x, u\}$ induces a $K_4-e$. So \ref{bto} holds.
       \end{proof}

\section{The class of ($2P_1+P_3, K_4-e$)-free graphs}\label{sec:2P1+P3}

In this section, we give a proof of \cref{col-thm-2}. First we give some important structural properties of  ($2P_1+P_3, K_4-e$)-free graphs  that  contain a  $C_5$  in \cref{genprop-3}, and use them   to prove a  structure theorem  for the class of ($2P_1+P_3, K_4-e$)-free graphs  in   \cref{Sec:Struc-thm2}. Finally we prove \cref{col-thm-2} in  \cref{col-2P1P3K4-e}.  Throughout this section, unless otherwise stated, the arithmetic operations on indices are understood to be modulo $5$.

 \subsection{Properties of  ($2P_1+P_3, K_4-e$)-free graphs  that  contain a  $C_5$}\label{genprop-3}

Let $G$ be a ($2P_1+P_3, K_4-e$)-free graph.
Suppose that $G$ contains a  $C_5$, say with
 vertex-set $C:= \{v_1, v_2, v_3, v_4, v_5\}$ and edge-set $\{v_1 v_2, v_2 v_3, v_3 v_4, v_4 v_5, $ $v_5 v_1\}$. Then, with respect to $C$, we define the sets $A$, $B$, $D$, $Z$, $X$ and $T$ as in Section~\ref{genprop}, and we use  the properties given in \cref{genprop}.
 In addition, for each $i\in [5]$, the following hold:
\begin{enumerate}[label=  $(\mathbb{M}\arabic*)$, leftmargin=1.25cm] \itemsep=0pt
\item \label{cliq}    {\it  $A_i\cup T$ is a clique.}
\item \label{ato-2}   {\it  $[A_i, (A\sm A_i)\cup   (D\sm D_{i}) \cup X_{i+2}]$ is complete.}
 \item \label{bd}  {\it One of $X_i$ and $D_{i-1} \cup D_{i+2}$ is  empty.}
  \end{enumerate}

\begin{proof}
\ref{cliq}:~If there are nonadjacent vertices in $A_i \cup T$, say $u$ and $v$, then $\{u, v, v_{i+1}, v_{i+2}, v_{i-2}\}$ induces a $2P_1+P_3$. So \ref{cliq} holds. $\sq$

\smallskip
\noindent{\ref{ato-2}}:~By symmetry, it is enough to prove that $[A_i, A_{i+1}\cup A_{i+2}\cup D_{i+1}\cup D_{i+2}\cup X_{i+2}]$ is complete. If there are nonadjacent vertices, say $a\in A_i$ and   $u \in A_{i+1}\cup A_{i+2}\cup D_{i+1}\cup D_{i+2}\cup X_{i+2}$, then $\{a, v_{i-1}, u, v_{i+1}, v_{i+2}\}$ induces a $2P_1+P_3$.  So \ref{ato-2} holds. $\sq$

\smallskip
\noindent{\ref{bd}}:~Suppose to the contrary that there are vertices, say $x \in X_i$ and $d \in D_{i-1} \cup D_{i+2}$. Then from \ref{bto}, we have $xd\notin E(G)$.   But then $\{d, v_{i-1}, x, v_{i+1}, v_{i+2}\}$ or   $\{d, v_{i+2}, x, v_{i}, v_{i-1}\}$ induces a $2P_1+P_3$  which is a contradiction. So \ref{bd} holds.
  \end{proof}

\subsection{Structure of ($2P_1+P_3, K_4-e$)-free graphs} \label{Sec:Struc-thm2}
Our main result in this subsection is the following.
 \begin{theorem}\label{main-thm-2}
Let $G$ be a  ($2P_1+P_3, K_4-e$)-free graph which is not perfect. Then     $G$ has a vertex of degree  at most 5 or $G$ is a good graph with $\omega(G)\geq 4$ or    $G$ is $6$-colorable or $G\in \cal C$.
\end{theorem}

The proof of \cref{main-thm-2} follows from the lemmas given below, and it is given at the end of this subsection.

\begin{lemma}\label{l1}
    Let $G$ be ($2P_1+P_3, K_4-e$)-free graph. If $G$ contains a $C_5+K_1$, then $G$ has a vertex of degree at most 5.
\end{lemma}
\begin{proof}
Let $G$ be a ($2P_1+P_3, K_4-e$)-free graph. Suppose that $G$ contains a $C_5 + K_1$, say a $C_5$ with
vertex-set $C := \{v_1, v_2, v_3, v_4, v_5\}$ and edge-set \{$v_1 v_2$, $v_2 v_3$, $v_3 v_4$, $v_4 v_5$, $v_5 v_1$\} plus a vertex $t$ which is anticomplete to  $C$. Then, with respect to $C$, we define the sets $A$, $B$, $D$, $Z$, $X$ and $T$ as in  \cref{genprop}, and we use the properties given in \cref{genprop} and \cref{genprop-3}. Clearly  $t \in T$. In addition, we claim that the following hold:

\begin{claim}\label{tto-1}
    $[T,  B \cup D \cup Z]$ is complete.
\end{claim}
\no{\it Proof of \cref{tto-1}}.~Suppose to the contrary that there are nonadjacent vertices, say $t'\in T$ and $u\in  B\cup D \cup Z$. Then we may assume that $u\in X_i\cup D_{i+2}$, for some $i\in [5]$. But then $\{t',v_{i-1},u,v_{i+1},v_{i+2}\}$ induces a $2P_1+P_3$ which is a contradiction. So \cref{tto-1} holds. $\sq$

\begin{claim}\label{t-card}
For each $i\in [5]$, we have $|A_i\cup Z_{i+2}|\leq 1$,  $|X_i\cup D_{i+2}|\leq 1$ and $|X_i\cup D_{i-1}|\leq 1$.
\end{claim}
\no{\it Proof of \cref{t-card}}.~If there are two vertices in  $A_i\cup Z_{i+2}$,   say $u$ and $u'$, then from \ref{DZ}, \ref{cliq}, \ref{ato-2} and from \cref{tto-1}, we see that $\{t,u,u',v_i\}$ induces a $K_4-e$; so we have $|A_i\cup Z_{i+2}|\leq 1$. Next if there are two vertices in  $X_i\cup D_{i+2}$,  say $v$ and $v'$, then from    \cref{tto-1}, we see that $\{t,v,v',v_{i+1}\}$ induces a $K_4-e$ or $\{v_{i-1}, v_{i+2},v,t,v'\}$ induces a $2P_1+P_3$; so  we have $|X_i\cup D_{i+2}|\leq 1$. Likewise,  we have $|X_i\cup D_{i-1}|\leq 1$. This proves \cref{t-card}. $\sq$

\medskip
So from the above claims,  we observe that $deg(v_1) = |N(v_1)| = |\{v_{2},v_{5}\}|+|X_5\cup D_2|+|X_1\cup D_5|+|A_1\cup Z_3| \leq 2 + 1 + 1 +1 = 5$. This proves Lemma \ref{l1}.
\end{proof}

\begin{lemma}\label{lem-2k1p3-c5}
    Let $G$ be ($2P_1+P_3, K_4-e$)-free graph. If $G$ contains a $C_5$, then $G$ has a vertex of degree at most 5 or $G$ is a good graph with $\omega(G)\geq 4$ or $G$ is $6$-colorable or $G\in \cal C$.
\end{lemma}
\begin{proof}
Let $G$ be a ($2P_1+P_3, K_4-e$)-free graph. Suppose that $G$ contains a $C_5$, say with
vertex-set $C := \{v_1, v_2, v_3, v_4, v_5\}$ and edge-set \{$v_1 v_2$, $v_2 v_3$, $v_3 v_4$, $v_4 v_5$, $v_5 v_1$\}. Then, with respect to $C$, we define the sets $A$, $B$, $D$, $Z$, $X$ and $T$ as in  \cref{genprop} and we  use the properties given in \cref{genprop} and \cref{genprop-3}.  By \cref{l1}, we may assume that $G$ is $(C_5+K_1)$-free; so $T=\es$.  In addition:

\begin{claim}\label{2.1}
        We may assume that $A$ is an empty set.
    \end{claim}
    \noindent{\it Proof of \cref{2.1}}.~Suppose that $A\neq \es$. By symmetry, we may assume that $A_1 \neq \emptyset$, and let $a_1\in A_1$. Then the following hold:

\begin{subclaim}\label{A-2k1p3-X15}  We may assume that $X_1\cup X_5$ is an empty set.
\end{subclaim}
\no{\em Proof of \cref{A-2k1p3-X15}}.~If there is a vertex in $X_1\cup X_5$, say $x$, then $\{a_1,v_5,x,v_2,v_3\}$ or $\{a_1,v_2,x,v_5,v_4\}$ induces a $2P_1+P_3$ (by \ref{bto}); so we may assume that $X_1\cup X_5=\es$. $\diamond$

\begin{subclaim}\label{A-2k1p3-D25-card}
$[D_2, D_5]$ is complete, $|D_2|\leq 1$ and $|D_5|\leq 1$.
\end{subclaim}
\no{\em Proof of \cref{A-2k1p3-D25-card}}.~If there are nonadjacent vertices, say $d_2\in D_2$ and $d_5\in D_5$, then  $\{d_2,v_2, d_5,v_4,v_5\}$   induces a $2P_1+P_3$; so $[D_2, D_5]$ is complete. Also  if there are two vertices  in $D_2$, say $d_2$ and $d_2'$, then from \ref{ato-2} and \ref{DZ}, $\{v_2,v_4,d_2,a_1,d_2'\}$ induces a $2P_1+P_3$; so we have $|D_2|\leq 1$. Likewise, we have $|D_5|\leq 1$. $\diamond$

\begin{subclaim}\label{A-2k1p3-D25-A1}
We may assume that $(A\sm A_1)\cup B_3\cup D_1\cup D_3\cup D_4$ is an empty set.
\end{subclaim}
\no{\em Proof of \cref{A-2k1p3-D25-A1}}.~First if  $|A_1|=1$, then from \ref{A-2k1p3-X15}, \ref{A-2k1p3-D25-card}, \ref{DZ} and \ref{bd}, we note that  $deg(v_1) = |N(v_1)| = |\{v_2,v_5\}|+|A_1|+|D_2 \cup D_5 \cup Z_3| \leq 2 + 1 + 2 = 5$,   and we are done; so we may assume that $|A_1|\geq 2$, and  we let $a_1'\in A_1\sm \{a_1\}$.  Recall that $a_1a_1'\in E(G)$, by \ref{cliq}. Now if there is a vertex in $(A\sm A_1)\cup B_3\cup D_3\cup D_4$, say $u$, then  $\{u,a_1,a_1',v_1\}$ induces a $K_4-e$ (by  \ref{ato-2}); so we have $(A\sm A_1)\cup B_3\cup D_3\cup D_4=\es$. Also if there is a vertex in $D_1$,  say $d_1$, then  $\{d_1,v_2,v_3,v_4,v_5,a_1\}$ or $\{d_1,v_2,v_3,v_4,v_5,a_1'\}$ induces a $C_5+K_1$ or $\{d_1,a_1,a_1',v_1\}$ induces a $K_4-e$; so we have   $D_1=\es$. $\diamond$

\begin{subclaim}\label{A-2k1p3-ome}
We may assume that $\omega(G)\geq 4$.
\end{subclaim}
\no{\em Proof of \cref{A-2k1p3-ome}}.~If $|X_2| \leq 1$, then from \ref{A-2k1p3-X15}, \ref{A-2k1p3-D25-A1} and \ref{DZ}, we have $deg(v_2) = |N(v_2)| =|\{v_1, v_3\} \cup X_2 \cup Z_4| \leq 2+1+1=4$ and we are done; so we may assume that $|X_2|\geq 2$. Then since $\{v_2,v_3\}\cup X_2$ is a clique (by \ref{X-matching}), we have  $\omega(G)\geq 4$. $\diamond$

\begin{subclaim}\label{A-2k1p3-clq}
 The set $\{v_1\} \cup A_1 \cup D_2\cup D_5\cup Z_{3}$ is a clique.
\end{subclaim}
\no{\em Proof of \cref{A-2k1p3-clq}}.~Since one of $Z_{3}$ and $D_2\cup D_5$ is empty (by \ref{bd}), clearly from  \ref{cliq}, \ref{ato-2}, \ref{DZ} and from \ref{A-2k1p3-D25-card}, we see that        $\{v_1\} \cup A_1 \cup D_2\cup D_5\cup Z_{3}$ is a clique. $\diamond$

\medskip
To proceed further, we let $Q_1 := \{v_1\} \cup A_1 \cup D_2\cup D_5\cup Z_{3}$, $Q_2 := \{v_2, v_3\} \cup X_2$, and $Q_3 := \{v_4, v_5\} \cup X_4$, and we claim the following.

\begin{subclaim}\label{A1graded}
The sets $Q_1$, $Q_2$ and $Q_3$ are pairwise graded.
\end{subclaim}
\no{\it Proof of  \cref{A1graded}}.~By \ref{X-matching} and by symmetry, it is enough to prove that  $\{Q_1, Q_2\}$ is graded. Now by  \cref{A-2k1p3-clq} and \ref{X-matching}, we know that $Q_1$ and $Q_2$ are cliques, and  so it is enough to prove that $[Q_1,  Q_2]$ is special.

        First we show that each vertex in $Q_1$ is adjacent to at most one vertex in $Q_2$. Suppose to the contrary that there is a vertex  in $Q_1$, say $q_1$, such that $|N(q_1)\cap Q_2|\geq 2$. Let $q_2,q_2'\in N(q_1)\cap Q_2$. Then clearly $q_1\notin \{v_1\}\cup D_2\cup Z_3$, by \ref{bto}; so $q_1\in A_1\cup D_5$. Then since $A_1\cup D_5$ is anticomplete to $\{v_2,v_3\}$, clearly $\{q_2,q_2'\}\subseteq X_2$. But then $\{q_1,q_2,q_2',v_3\}$ induces a $K_4-e$ which is a contradiction. So each vertex in $Q_1$ is adjacent to at most one vertex in $Q_2$.

    Next we show that each vertex in $Q_2$ is adjacent to at most one vertex in $Q_1$. Suppose to the contrary that there is a vertex in $Q_2$,   say $q_2$, such that $|N(q_2)\cap Q_1|\geq 2$. Let $q_1,q_1'\in N(q_2)\cap Q_1$. Clearly $q_2\neq v_2$, and  since $|N(v_3)\cap Q_1| \leq |D_2\cup Z_3|\leq 1$ (by \ref{A-2k1p3-D25-card}, \ref{DZ} and \ref{bd}), we have $q_2\neq v_3$. So $q_2\in X_2$. Then since $Q_1\sm \{v_1\}$ is complete to $\{v_1\}$, we see that $\{q_2,q_1,q_1',v_1\}$ induces a induces a $K_4-e$  which is a contradiction. So each vertex in $Q_2$ is adjacent to at most one vertex in $Q_1$. This proves \cref{A1graded}. $\diamond$

    \medskip
        Thus from \ref{A-2k1p3-X15} and \ref{A-2k1p3-D25-A1}, we observe that $V(G) = C \cup A_1 \cup D_2 \cup D_5 \cup X_2 \cup X_4 \cup Z_{3} =Q_1\cup Q_2\cup Q_3$. Then from  \ref{X-matching}, \ref{A-2k1p3-ome} and \ref{A1graded}, we see that $Q_1$, $Q_2$ and $Q_3$  are three cliques such that $Q_i$'s are pairwise graded  with $\omega(G) \geq 4$.   Hence $G$ is a good graph with $\omega(G)\geq 4$, and we  are done. So we may assume that $A=\es$. $\sq$

Next:

     \begin{claim}\label{2.2}
        We may assume that $B$ is an empty set.
    \end{claim}
\noindent{\it Proof of \cref{2.2}}.~Suppose that $B\neq \es$.  If $B_i\neq \es$, then we let $b_i$ denote a vertex in $B_i$. First if there is an $i\in [5]$ such that $B_{i-1},B_i$ and $ B_{i+1}$ are nonempty, then by symmetry, we may assume that $i=1$, and then from \ref{bto},   $\{b_1, b_2, b_5, v_5, v_4\}$   induces a   $2P_1+P_3$ (when $b_2 b_5 \notin E(G)$) or $\{b_2, v_3,v_4,v_5,b_5, b_1\}$ induces a $C_5 + K_1$ (when $b_2 b_5 \in E(G)$) which is a contradiction; so  for each  $i\in [5]$, at least one of  $B_{i-1},B_i$ and $ B_{i+1}$  is empty.

    Next suppose that there is an index $i\in [5]$ such that  $B_i$ and $ B_{i+1}$ are nonempty. By symmetry, we may assume that $i=1$. Then $B_3\cup B_5=\es$, and $D_1\cup D_3\cup D_4\cup D_5 =\es$ (by \ref{bd}). Also if there is a vertex in $D_2$,  say $d_2$,  then $\{d_2, v_4, b_1, v_2, b_2\}$  induces a  $2P_1+P_3$ (by \ref{bto}); so    we have $D_2=\es$. Moreover, if there is a vertex in $Z_3\cup Z_5$,  say $z$, then
    $\{b_1,b_2, z,v_4,v_5\}$  induces a  $2P_1+P_3$ (by \ref{bto}); so we have $Z_3\cup Z_5=\es$.
    Hence from \cref{2.1}, $V(G) = C \cup X_1 \cup X_2 \cup X_4$.     Recall that $[X_1, X_2] = \emptyset$ (by \ref{bto}), and $\{X_1, X_4\}$ and $\{X_2, X_4\}$ are graded (by \ref{X-matching}). Now if $|X_1|\leq 2$, then $deg(v_1)= |N(v_1)| = |\{v_2,v_5\} \cup X_1| \leq 4$, and we are done; so we may assume that $|X_1|\geq 3$. Similarly, we may assume that $|X_2|\geq 3$.  Moreover, if $|X_4|\leq 2$, then $deg(v_4)= |N(v_4)| = |\{v_3,v_5\} \cup X_4\cup Z_1| \leq 5$ (by \ref{DZ}), and we are done; so we may assume that $|X_4|\geq 3$.    Then it is easy to see that $G \in {\cal C}$ by taking $u_j:=v_{j+1}$, for $j\in [5]$ and $j$ mod $5$,  $Q_1:=B_2$, $Q_1':=Z_2$, $Q_2:=B_4$, $Q_2':=Z_4$, $Q_3:=B_1$ and $Q_3':=Z_1$, and we are done.

\smallskip
 From now on, we may assume that for each  $i\in [5]$, at least one of  $B_i$ and $ B_{i+1}$  is empty. Also since  $B\neq \es$, we may assume the $B_1 \neq \emptyset$.
      If $B\sm B_1=\es$, then from  \cref{2.1}, \ref{bd} and \ref{DZ},  we have $deg(v_4) = |N(v_4)| = |\{v_3,v_5\} \cup Z_1 \cup Z_3 \cup Z_4| \leq 5$, and we are done; so we may assume that $B_3 \neq \emptyset$.  Then  since $B_1\neq \es$ and $B_3\neq \es$, we have $B_2 \cup B_4 \cup B_5 = \emptyset$  and $D_2 \cup D_3 \cup D_5 = \emptyset$ (by \ref{bd}).
   Also if  there are two vertices in $D_1$,  say $d_1$ and $d_1'$, then $\{b_1, v_3, d_1, v_5, d_1'\}$  induces a $2P_1+P_3$ (by  \ref{DZ} and \ref{bto}); so we have $|D_1|\leq 1$.  Similarly, we have $|D_4| \leq 1$. Hence $|D_1\cup Z_2\cup Z_4|\leq 2$ and $|D_4\cup Z_5|\leq 1$ (by  \ref{DZ} and \ref{bd}). Then from \cref{2.1}, we see that $deg(v_5) = |N(v_5)| = |\{v_1,v_4\}|+|D_1 \cup  Z_2 \cup Z_4|+|D_4 \cup Z_5|  \leq 2+2+1 =5$, and we are done.  This proves \ref{2.2}. $\sq$

  \medskip
    Now  from \cref{2.1} and \cref{2.2}, we conclude that $V(G) = C \cup  D \cup Z$. Now we let $S_j:= \{v_j\}\cup D_j$, for $j\in [5]$, and let $S_6:= Z$.    Then from  \ref{DZ} and \ref{bto}, we conclude that $S_1, S_2, \ldots, S_6$ are stable sets, and hence $G$ is $6$-colorable. This proves \cref{lem-2k1p3-c5}.
\end{proof}

\begin{lemma}\label{lem-2k1p3-c7}
If $G$ is a ($2P_1+P_3, K_4-e, C_5$)-free graph that contains a $C_7$, then $G$ is isomorphic to $C_7$.
\end{lemma}
\begin{proof}
Let $G$ be a ($2P_1+P_3, K_4-e, C_5$)-free graph. Suppose that $G$ contains a $C_7$, say with vertex-set $C := \{v_1, v_2, v_3, v_4, v_5, v_6, v_7\}$ and edge-set \{$v_1 v_2$, $v_2 v_3$, $v_3 v_4$, $v_4 v_5$, $v_5 v_6$, $v_6 v_7$, $v_7 v_1$\}.    Suppose to the contrary that $V(G)\sm C\neq \es$, and let $u\in V(G)\sm C$. Since $C\cup \{u\}$ does not induce a $2P_1+P_3$, clearly $u$ has a neighbor in $C$. Then we claim the following.
\begin{equation}\label{lem-2k1p3-c7-claim}
  \longbox{\em We may assume that for each $i\in [7]$ and $i$ mod $7$,   $u$ is nonadjacent to at least one of $v_i$ and $v_{i+1}$.} \tag*{($\star$)}
   \end{equation}
 \no{\it Proof of} \ref{lem-2k1p3-c7-claim}.~Suppose to the contrary that there is an index $i\in [7]$ such that $uv_i, uv_{i+1}\in E(G)$. By symmetry, we may assume that $i=1$.  Now if $uv_3\in E(G)$, then $\{v_1,v_2,v_3,u\}$ induces a $K_4-e$; so we have $uv_3\notin E(G)$. Similarly, we have $uv_7\notin E(G)$. Then $uv_5\in E(G)$, for otherwise, $\{v_3,v_5,v_7,v_1,u\}$ induces a $2P_1+P_3$. Hence  $uv_4\in E(G)$, for otherwise, $\{v_2,v_3,v_4,v_5,u\}$ induces a $C_5$. Similarly, we have $uv_6\in E(G)$. But then $\{u,v_4,v_5,v_6\}$ induces a $K_4-e$ which is a contradiction. So  \ref{lem-2k1p3-c7-claim} holds. $\sq$

\medskip
 Now since $u$ has a neighbor in $C$, we may assume (without loss of generality) that $uv_1\in E(G)$. Then by \ref{lem-2k1p3-c7-claim}, we have $uv_2,uv_7\notin E(G)$. Thus  if $uv_4\in E(G)$, then by \ref{lem-2k1p3-c7-claim}, we have $uv_3\notin E(G)$ and then  $\{u,v_1,v_2,v_3,v_4\}$ induces a $C_5$; so we have $uv_4\notin E(G)$, and similarly,  we have $uv_5\notin E(G)$.   Then since $\{v_3,v_5, u,v_1,v_7\}$ and $\{v_4,v_6, u,v_1,v_2\}$   do not induce ($2P_1+P_3$)'s, we have $uv_3,uv_6\in E(G)$.  But then $\{u,v_3,v_4,v_5,v_6\}$ induces a $C_5$ which is a contradiction. This proves \cref{lem-2k1p3-c7}.
\end{proof}

\noindent{\bf Proof of \cref{main-thm-2}}. Let $G$ be a ($2P_1+P_3, K_4-e$)-free graph which is not perfect.   Then by \cref{spgt}, $G$ contains  an odd-hole or the complement graph of an odd-hole. Since the
complement graph of any odd-hole of length at least 7 contains a $K_4-e$, and since any odd-hole of length at least 9 contains a $2P_1+P_3$, we may assume that $G$ contains   a $C_5$ or a $C_7$. If $G$ is $C_5$-free and $G$ contains a $C_7$, then from \cref{lem-2k1p3-c7}, $G\cong C_7$, and hence $G$ is $3$-colorable and we are done. So suppose that $G$ contains a $C_5$. Then the theorem follows from  \cref{lem-2k1p3-c5}. This completes the proof of \cref{main-thm-2}. \hfill{$\Box$}

\subsection{Coloring of ($2P_1+P_3, K_4-e$)-free graphs} \label{col-2P1P3K4-e}

In this subsection, we prove \cref{col-thm-2} and its consequences.  


\medskip

\noindent{\bf Proof of \cref{col-thm-2}}.~Let $G$ be a  ($2P_1+P_3, K_4-e$)-free graph with $\omega(G)\geq 3$. First we prove \cref{col-thm-2}:$(i)$. Since $\omega(G)=3$, we may assume  that $G$ is $K_4$-free.  Let $v$ be any vertex in $G$.
 Then   clearly $G[N(v)]$ is  a ($K_3, P_3$)-free graph which is a  bipartite graph, and   $G[\overline{N}(v)]$ is  a ($P_1+P_3, K_4-e, K_4$)-free graph.
It follows from a result of Gei\ss er \cite{Geiber} that  $\chi(G[\overline{N}(v)])\leq 3$. So $\chi(G)\leq \chi(G[N(v)])+\chi(G[\overline{N}(v)\cup\{v\}])\leq 2+3 =5$, and we are done. This proves \cref{col-thm-2}:$(i)$.

Next we prove \cref{col-thm-2}:$(ii)$. Let $G$ be a ($2P_1+P_3, K_4-e$)-free graph with $\omega(G)\geq 4$, and let $\ell = \max\{6, \omega(G)\}$. We may assume that $G$ is not perfect. We prove   \cref{col-thm-2}:$(ii)$ by induction on $|V(G)|$.  From   \cref{main-thm-2},  \cref{lem-good} and from  \cref{classC-col}, we may assume that   $G$ has a vertex of degree at most $5$, say $v$.    By induction hypothesis, we have $\chi(G - \{v\}) \leq  \ell$. Since $v$ has at most 5 neighbors in $G - \{v\}$, we can take any $\ell$-coloring of $G - \{v\}$ and extend it to a $\ell$-coloring of $G$, using for $v$ a (possibly new) color that does not appear in its neighborhood.  This proves \cref{col-thm-2}:$(ii)$.

The bound  is clearly tight when $\omega(G)\geq 6$ (for instance, take $G\cong K_t, t\geq 6$).   \hfill{$\Box$}

\medskip
We remark that   the bound given in \cref{col-thm-2} does not seem to be tight when $\omega(G)\in \{3,4,5\}$. Note that there are ($2P_1+P_3, K_4-e$)-free graphs with chromatic number equal to $4$ and clique number equal to 3; see  $H_1$ and $H_2$ in Figure~\ref{fig-opt}.

\begin{cor}\label{2P1P3-NOC}
The class of ($2P_1+P_3, K_4-e$)-free graphs is near optimal colorable. That is, every ($2P_1+P_3, K_4-e$)-free graph $G$ satisfies $\chi(G)\leq \max\{6, \omega(G)\}$.
\end{cor}
\begin{proof}
The proof follows from the fact that every ($2P_1+P_3, K_3$)-free graph $G$ satisfies $\chi(G)\leq 3$ \cite{Rand-Thesis} and from \cref{col-thm-2}.
\end{proof}

 Brandst\"adt et al. \cite{BDHP} showed that \textsc{Coloring} for  the class of ($2P_1+P_3$, $K_4-e$)-free graphs is solvable in  polynomial time   for every fixed positive integer $k \leq 5$. This result together with \cref{2P1P3-NOC} and \cref{JH-thm}  imply  the following.

 \begin{cor}
  \textsc{Chromatic Number} for
	the class of ($2P_1+P_3$, $K_4-e$)-free graphs can be solved in  polynomial time.
 \end{cor}


\section{The class of ($3P_1+P_2$,\,$K_4-e$)-free graphs}\label{sec:3P1+P2}

In this section, we give a proof of \cref{col-thm}.  In \cref{3P1P2-prop}, we give some important structural properties of  ($3P_1+P_2$,\,$K_4-e$)-free graphs that contain a $C_5$, and use them   to prove a  structure theorem  for the class of ($3P_1+P_2$,\,$K_4-e$)-free graphs in   \cref{Sec:Struc-thm}. Finally using the structure theorem, we prove \cref{col-thm} in  \cref{Sec:col}. Throughout this section, unless otherwise stated,   the arithmetic operations on indices are understood to be modulo $5$.

\subsection{Properties of ($3P_1+P_2$,\,$K_4-e$)-free graphs that contain a $C_5$}\label{3P1P2-prop}
Let $G$ be a ($3P_1+P_2$,\,$K_4-e$)-free graph.
Suppose that $G$ contains a  $C_5$, say with
 vertex-set $C:= \{v_1, v_2, v_3, v_4, v_5\}$ and edge-set $\{v_1 v_2, v_2 v_3, v_3 v_4, v_4 v_5, $ $v_5 v_1\}$. Then, with respect to $C$, we define the sets $A$, $B$, $D$, $Z$, $X$ and $T$ as in Section~\ref{genprop}, and we  use the properties given in \cref{genprop}.   Moreover
for each $i\in [5]$, the following hold:
{\it \begin{enumerate}[label=  $(\mathbb{L} \arabic*)$, leftmargin=1.25cm] \itemsep=0pt
\item\label{df} $|D_{i}|\leq 1$.
\item \label{abf}       $A_{i}\cup T$ is a clique.
\item\label{ato}  $[A_{i}, A_{i-2} \cup A_{i+2} \cup  D_{i-1}\cup D_{i+1}]$ is complete.
\item\label{abz} If $A_i\cup D_{i+1}\neq \es$, then $|X_{i-1}|\leq 1$. Similarly, if $A_i\cup D_{i-1}\neq \es$, then $|X_i|\leq 1$.
\item\label{abb}  One of  $A_{i}\cup D_i$, $B_{i-1}$ and $B_{i}$ is   empty.	
\item\label{afd} One of $A_{i}$, $D_{i+1}\cup D_{i-1}$ and $Z_{i+2}$ is   empty.
\item\label{bfd}One of $B_i$, $D_{i+2} \cup D_{i-1}$ and $Z_{i}$ is empty.
\item \label{zzd} One of $Z_{i+2}$, $B_i \cup B_{i+1} \cup Z_{i}$ and $D_{i+1}$ is empty. Moreover, $|D_{i+1}\cup Z_{i}\cup Z_{i+2}|\leq 2$.
  \end{enumerate}
}
        \begin{proof}
        \ref{df}:~If there are two vertices in $D_i$, say $d$ and $d'$, then $dd'\notin E(G)$ (by  \ref{DZ}), and then $\{v_{i}, d, d', v_{i-2}, v_{i+2}\}$ induces a  $3P_1+P_2$. So \ref{df} holds. $\sq$

\smallskip
  \no{\ref{abf}}:~If there are nonadjacent vertices  in $A_{i}\cup T$, say $u$ and $u'$, then $\{u,u',v_{i-1},v_{i+1}, v_{i+2}\}$ induces a $3P_1+P_2$. So  \ref{abf} holds. $\sq$

\smallskip
  \no{\ref{ato}}:~By symmetry, it is enough to prove that $[A_{i}, A_{i+2} \cup D_{i+1}]$ is complete. Now if  there are nonadjacent vertices, say $a\in A_{i}$ and $u\in A_{i+2}\cup D_{i+1}$, then $\{a,u,v_{i+1},v_{i-1},v_{i-2}\}$ induces a $3P_1+P_2$. So \ref{ato} holds.  $\sq$

\smallskip
  \no{\ref{abz}}:~Let $u\in A_i\cup D_{i+1}$. Now if there  are two vertices in $X_{i-1}$, say $x$ and $x'$, then from \ref{X-matching}, we have $xx'\in E(G)$, and then we see that $\{u, v_{i+1},v_{i-2},x,x'\}$ induces a $3P_1+P_2$, by \ref{bto}. So \ref{abz} holds. $\sq$

\smallskip
  \no{\ref{abb}}:~Suppose to the contrary that there are vertices, say $u\in A_{i}\cup D_i$, $b\in B_{i-1}$ and  $b'\in B_{i}$. Then from \ref{bto}, clearly $\{u,b,b'\}$ is a stable set. But then   $\{u,b,b',v_{i-2}, v_{i+2}\}$ induces a $3P_1+P_2$ which is a contradiction. So \ref{abb} holds. $\sq$

\smallskip
  \no{\ref{afd}}:~Suppose to the contrary that there are vertices, say $a \in A_i$, $d\in D_{i+1}\cup D_{i-1}$ and $z\in Z_{i+2}$. Then $ad\in E(G)$, by \ref{ato} and $dz \not\in E(G)$, by \ref{bto}. Now if $az \in E(G)$, then     $\{d,a,v_i,z\}$ induces a $K_4-e$; so $az\notin E(G)$. But then $\{v_{i-1}, v_{i+1},  z,  a, d\}$  induces a $3P_1+P_2$ which is a contradiction. So \ref{afd} holds. $\sq$

\smallskip
  \no{\ref{bfd}}:~Suppose to the contrary that there are vertices, say $b \in B_i$, $d\in D_{i+2}\cup D_{i-1}$  and $z\in Z_{i}$. Then using \ref{X-matching} and \ref{bto}, we see that $\{ d, v_{i+2}, v_{i-1}, b, z\}$ induces $3P_1+P_2$ which is a contradiction. So \ref{bfd}  holds.  $\sq$

\smallskip
  \no{\ref{zzd}}:~Suppose to the contrary that there are vertices, say $z \in Z_{i+2}$, $u\in B_i \cup B_{i+1} \cup Z_{i}$  and $d\in D_{i+1}$. Then from \ref{bto}, we see that $\{ v_{i-1}, z, d, u, v_{i+1}\}$ induces $3P_1+P_2$ which is a contradiction.
So the first assertion holds. The second assertion follows from the first assertion, \ref{DZ} and \ref{df}.
         \end{proof}


\subsection{Structure of ($3P_1+P_2, K_4-e$)-free graphs} \label{Sec:Struc-thm}
Our main result in this subsection is the following.
 \begin{theorem}\label{main-thm}
If $G$ is a  ($3P_1+P_2, K_4-e$)-free graph, then     $G$ has a vertex of degree  at most 6 or $G$ is a good graph with $\omega(G)\geq 4$ or $G\in \cal C$.
\end{theorem}

The proof of \cref{main-thm} is based on a sequence of lemmas given below which depend on some special forbidden induced subgraphs that contain a $C_5$ or a $C_7$ (some of them are given in Figure~\ref{fig-F1234}), and it is given at the end of this subsection. We start with the following.

\begin{lemma}\label{NC5:deg6}
   Let $G$ be a  ($3P_1+P_2, K_4-e$)-free graph. If $G$ contains  a $C_5+K_1$, then $G$ has a vertex of degree  at most 6.
\end{lemma}
\begin{proof}
 Let $G$ be a  ($3P_1+P_2, K_4-e$)-free graph. Suppose that $G$ contains  a $C_5+K_1$,   say a $C_5$ with
 vertex-set $C:= \{v_1, v_2, v_3, v_4, v_5\}$ and edge-set $\{v_1 v_2, v_2 v_3, v_3 v_4, v_4 v_5, $ $v_5 v_1\}$ plus a vertex $t$  which is anticomplete to  $C$. Then, with respect to $C$, we define the sets $A$, $B$, $D$, $Z$, $X$ and $T$ as in Section~\ref{genprop}, and we use the properties in Section~\ref{genprop} and in \cref{3P1P2-prop}.  Clearly $t\in T$. Moreover:

\begin{claim}\label{tto}
 $[T, D]$ is complete. Moreover, for each $i\in [5]$,  we have  $|A_i|\leq 1$ and $|X_i|\leq 2$.
\end{claim}
   \no{\it Proof of  \cref{tto}}.~To prove that $[T, D]$ is complete, it is enough to prove that $[T, D_{i}]$ is complete for each $i\in [5]$. Now if  there are nonadjacent  vertices, say $t'\in T$ and $d\in D_{i}$, then $\{t', d, v_{i}, v_{i-2},v_{i+2}\}$ induces a $3P_1+P_2$; so the first assertion holds.

   Next if there are two vertices  in $A_i$, say $a$ and $a'$, then   $\{t, a,a', v_i\}$ induces a $K_4-e$, by \ref{abf}; so we have $|A_i|\leq 1$.

Finally by symmetry, we prove the last assertion for $i=1$. Suppose to the contrary that there  are three vertices in $X_1$, say $x,x'$ and $x''$. Note that $\{x,x',x''\}$ is clique, by \ref{X-matching}. Now if $t$ is adjacent to at least two vertices in $\{x,x',x''\}$, then $G[\{t, x,x',x'', v_1\}]$  contains a $K_4-e$; so we may assume that $tx, tx'\notin E(G)$. Then $\{t,v_{3}, v_{5},x,x'\}$ induces a $3P_1+P_2$ which is a contradiction. So we have $|X_1|\leq 2$. So \ref{tto} holds. $\sq$

 Next:

\begin{claim}\label{TA-emp}
  We may assume that $A$ is an empty set.
  \end{claim}
  \no{\it Proof of  \cref{TA-emp}}.~Suppose that $A\neq \es$. By symmetry, we may assume that $A_1\neq \es$, and let $a_1\in A_1$.  Then $|X_1|\leq 1$ and $|X_5|\leq 1$, by \ref{abz}. Also if there is a vertex in $D_{2}\cup D_{5}$, say $d$, then $\{t,a_1,d,v_{1}\}$ induces a $K_4-e$ (by \ref{abf}, \ref{ato}  and  \ref{tto}); so we have $D_2\cup D_5=\es$. Hence $deg(v_1)=|N(v_1)| =
|\{v_2, v_5\}\cup A_1 \cup X_1\cup X_5 \cup  Z_3| = |\{v_2, v_5\}|+ |A_1|+ |X_1|+ |X_5|+ |Z_3|\leq 2+1+1+1+1 = 6$, by \ref{tto} and \ref{DZ}, and we are done; so we may assume that $A = \es$. $\sq$

 Next:

\begin{claim}\label{TD-emp}
  We may assume that $D$ is an empty set.
  \end{claim}
  \no{\it Proof of  \cref{TD-emp}}.~Suppose that $D\neq \es$. By symmetry, we may assume that $D_2\neq \es$, and let $d_2\in D_2$. Then the following hold:

  \begin{subclaim}\label{TD-B}
  One of $B_1$ and $B_5$ is empty.
  \end{subclaim}
  \no{\em Proof of \cref{TD-B}}.~If there are vertices, say $b_1\in B_1$ and $b_5\in B_5$,   then  $\{d_2, b_5, v_4, b_1,v_2\}$  induces a $3P_1+P_2$, by \ref{bto}; so one of $B_1$ and $B_5$ is empty. $\diamond$

  \begin{subclaim}\label{TD-Z13}
  One of $Z_3$ and $Z_5$ is empty.
  \end{subclaim}
  \no{\em Proof of \cref{TD-Z13}}.~If there are vertices, say  $z_3\in Z_3$ and $z_5\in Z_5$, then we observe that $\{d_2, v_2,z_3, z_5, v_5\}$ induces a $3P_1+P_2$, by  \ref{bto}; so one of $Z_3$ and $Z_5$ is empty. $\diamond$

  \begin{subclaim}\label{TD-Z4}
  One of $D_5$ and $Z_1$ is empty.
  \end{subclaim}
  \no{\em Proof of \cref{TD-Z4}}.~Suppose to the contrary that there are vertices, say $d_5\in D_5$ and $z_1\in Z_1$. Now if $d_2d_5\in E(G)$, then   $\{t,d_2,d_5,v_1\}$  induces a $K_4-e$ (by \ref{tto}); so we have $d_2d_5\notin E(G)$. Then   $\{v_5,d_5,z_1, v_3, d_2\}$  induces a $3P_1+P_2$ (by \ref{bto}) which is a contradiction; so one of $D_5$ and $Z_1$ is empty. $\diamond$

 \medskip
From \ref{TD-Z13} and \ref{DZ}, we have $|Z_3\cup Z_5|\leq 1$. From \ref{TD-Z4}, \ref{DZ} and  \ref{df}, we have $|D_5 \cup Z_1| \leq 1$. Moreover $|D_5\cup X_1| \leq 2$, by \ref{df}, \ref{abz} and \ref{tto}. Also $|D_2\cup X_5|\leq 2$, by  \ref{df}, \ref{abz} and \ref{tto}. Note that from  \cref{TA-emp}, $deg(v_1)=|N(v_1)| =|\{v_2, v_5\}\cup B_1 \cup B_5 \cup D_2 \cup D_5 \cup Z_1 \cup  Z_3 \cup Z_5|$. Now from \ref{TD-B}, we have two cases. If $B_5=\es$, then $deg(v_1)= |\{v_2, v_5\}|+  |D_2|+| Z_3\cup Z_5| + |D_5\cup X_1|\leq 2+1+1+2 = 6$, by \ref{df}, and we are done. So we may assume that $B_1=\es$. Then  $deg(v_1)= |\{v_2,v_5\}| + |D_2\cup X_5| + |D_5 \cup Z_1| + |Z_3|\leq 2 + 2 + 1 + 1 = 6$, by  \ref{DZ}, and again we are done. So we may assume that $D=\es$.  $\sq$

 Next:
\begin{claim}\label{TBi-Nemp}
  We may assume that $|X_i|=2$, for each $i\in [5]$.
  \end{claim}
\no{\it Proof of  \cref{TBi-Nemp}}.~Suppose that there is an index $i\in [5]$ such that $|X_{i}|\leq 1$. By symmetry, we may assume that $i=1$ and so $|X_{1}|\leq 1$. Then from \ref{DZ} and from the above claims, we see that $deg(v_1)=|\{v_2,v_5\}|+|X_1|+|X_5|+|Z_3|\leq 2+1+2+1 = 6$, and we are done. So from \cref{tto}, we may assume that $|X_i|=2$, for each $i \in [5]$. $\sq$

\medskip
  From \cref{TBi-Nemp},  for each $i\in [5]$, there are vertices in $X_i$, say $x_i$ and $x_i'$,  such that  $x_ix_i'\in E(G)$, by \ref{X-matching}. Then we claim the following.

 \begin{claim}\label{deg-t}
 For each $i\in [5]$, we have $|N(t)\cap X_i|\leq 1$. Moreover, $|N(t)\cap T|\leq 1$.
 \end{claim}
\no{\it Proof of  \cref{deg-t}}.~If  $|N(t)\cap X_i|=2$, then $\{t, x_i, x_i', v_i\}$ induces a $K_4-e$; so the first assertion holds. Next suppose to the contrary that there are vertices in $T$, say $t'$ and $t''$, such that $\{t,t',t''\}$ is a clique (by \ref{abf}). Then clearly $x_1$ is complete to  $\{t,t',t''\}$, for otherwise, $G[\{v_3,v_5,x_1,t,t',t''\}]$ contains a $3P_1+P_2$ or a $K_4-e$. Similarly, $x_2$ is complete to $\{t,t',t''\}$. But then from \ref{bto},  $\{x_1,t,t',x_2\}$ induces a $K_4-e$ which is a contradiction. So $|N(t)\cap T|\leq 1$. $\sq$

\medskip
From \ref{TA-emp}, \ref{TD-emp} and \cref{deg-t}, we conclude that $deg(t)\leq |N(t)\cap T|+ \sum_{i=1}^5|N(t)\cap X_i|\leq 6$. This proves \cref{NC5:deg6}.
\end{proof}


\begin{figure}[t]
\centering
 \includegraphics[height=3cm, width=13cm]{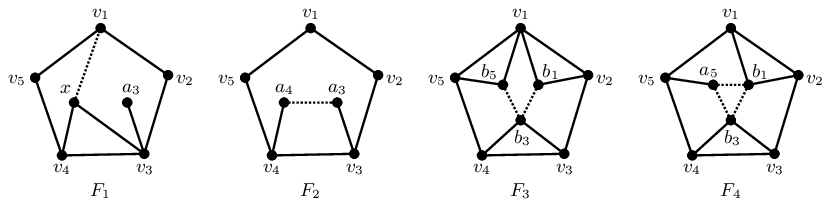}
\caption{Some special graphs. Here, a dotted line  between two vertices indicates that  those two vertices may or may not be adjacent.  }\label{fig-F1234}
\end{figure}

In the following,   we will refer to Figure~\ref{fig-F1234} for the graphs $F_1, F_2, F_3$ and $F_4$.

\begin{lemma}\label{F1:thm}
    Let $G$ be a  ($3P_1+P_2, K_4-e$)-free graph. If $G$ contains  an $F_1$, then $G$ has a vertex of degree  at most 6.
\end{lemma}
\begin{proof}
Let $G$ be a  ($3P_1+P_2, K_4-e$)-free graph. Suppose that $G$ contains  an $F_1$ as shown in Figure~\ref{fig-F1234}.  We let  $C:= \{v_1, v_2, v_3, v_4, v_5\}$. Then, with respect to $C$, we define the sets $A$, $B$, $D$, $Z$, $X$ and $T$ as in Section~\ref{genprop}, and we use the properties in Section~\ref{genprop}  and in \cref{3P1P2-prop}. Clearly $a_3\in A_3$ and $x\in X_3$.   Moreover  the following hold.

 \begin{claim}\label{F1-A3} We have $A_3=\{a_3\}$ and  so $|A_3|=1$.
  \end{claim}
  \no{\em Proof of \cref{F1-A3}}.~If there is a vertex in $A_3\sm \{a_3\}$, say $a_3'$, then $a_3a_3'\in E(G)$ (by \ref{abf}), and then from \ref{bto}, we see that $\{v_2,v_5,x,a_3,a_3'\}$ induces a $3P_1+P_2$; so $A_3=\{a_3\}$ and so $|A_3|=1$. $\sq$

Next:

   \begin{claim}\label{F1-D2}
    We may assume that $D_2$ is an empty set.
    \end{claim}
  \no{\em Proof of \cref{F1-D2}}.~If there is a vertex in $D_2$, say $d_2$, then $\{v_2,v_5,x,a_3,d_2\}$  induces a $3P_1+P_2$ (by \ref{bto} and \ref{ato}); so we may assume that $D_2=\es$. $\sq$

 \medskip
Now since $A_3\neq \es$,  we have $|D_4\cup Z_5|\leq 1$ (by \ref{DZ}, \ref{df} and \ref{afd}), and   $|X_2|\leq 1$ and $|X_3|=1$ (by \ref{abz}). So from the above arguments, we see that $deg(v_3)=|N(v_3)|= |\{v_2, v_4\}|+|A_3|+
 |X_2|+|X_3|+|D_4 \cup Z_5|\leq 2+1+1+1+1 =6$. This proves \cref{F1:thm}.
\end{proof}

\begin{lemma}\label{F2:thm}
    Let $G$ be a  ($3P_1+P_2, K_4-e$)-free graph. If $G$ contains  an $F_2$, then $G$ has a vertex of degree  at most 6.
\end{lemma}
\begin{proof}
Let $G$ be a  ($3P_1+P_2, K_4-e$)-free graph. From \cref{F1:thm}, we may assume that $G$ is $F_1$-free. Suppose that $G$ contains  an $F_2$ as shown in Figure~\ref{fig-F1234}.  We let  $C:= \{v_1, v_2, v_3, v_4, v_5\}$. Then, with respect to $C$, we define the sets $A$, $B$, $D$, $Z$, $X$ and $T$ as in Section~\ref{genprop}, and we use the properties in Section~\ref{genprop}  and in \cref{3P1P2-prop}. Clearly $a_3\in A_3$  and $a_4\in A_4$. Since $A_3, A_4\neq \es$  and since $G$ is $F_1$-free, we have $X_2\cup X_3=\es$. Also since $A_3\neq \es$, we have $|D_4\cup Z_5|\leq 1$ (by \ref{DZ}, \ref{df} and \ref{afd}). Moreover  if there are two vertices in $A_3\sm \{a_3\}$, say $a_3'$ and $a_3''$,  then since $\{a_3,a_3',a_3''\}$ is a clique (by \ref{abf})  and  since $G[\{v_2,v_5,a_4,a_3,a_3',a_3''\}]$ does not contain a $3P_1+P_2$, we see that $a_4$ is adjacent to at least two vertices in $\{a_3,a_3',a_3''\}$, and then $G[\{a_4, a_3,a_3',a_3'', v_3\}]$ contains a $K_4-e$; so we conclude that $|A_3|\leq 2$. Thus
   $deg(v_3)=  |\{v_2, v_4\}\cup A_3  \cup D_2\cup  D_4 \cup Z_5 |= |\{v_2, v_4\}|+|A_3|+
 |D_2|+ |D_4 \cup Z_5|\leq 2+2+1+1=6$, by \ref{df}. This proves \cref{F2:thm}.
\end{proof}

\begin{lemma}\label{F3:thm}
    Let $G$ be a  ($3P_1+P_2, K_4-e$)-free graph. If $G$ contains  an $F_3$, then  $G$ has a vertex of degree  at most 6  or $G\in \cal C$.
\end{lemma}
\begin{proof}
Let $G$ be a  ($3P_1+P_2, K_4-e$)-free graph.  From \cref{NC5:deg6,F1:thm}, we may assume that $G$ is ($F_1,  C_5+K_1$)-free. Suppose that $G$ contains  an $F_3$  as shown in Figure~\ref{fig-F1234}. We let  $C:= \{v_1, v_2, v_3, v_4, v_5\}$. Then, with respect to $C$, we define the sets $A$, $B$, $D$, $Z$, $X$ and $T$ as in Section~\ref{genprop}, and we use the properties in Section~\ref{genprop}  and in \cref{3P1P2-prop}. Clearly  $b_1\in B_1$, $b_3\in B_3$ and $b_5\in B_5$. Since $G$ is $F_1$-free, we have $A=\es$, and since $G$ is ($C_5+K_1$)-free, we have $T=\es$. Recall that $[B_1, B_5] =\es$, by \ref{bto}.  In addition, we have the following.
 \begin{claim}\label{F3-D125}
 The set $D_1\cup D_2\cup D_5$ is an empty set.
 \end{claim}
\no{\it Proof of  \cref{F3-D125}}.~By \ref{abb}, we have that $D_1=\es$. It remains to show that $D_2\cup D_5=\es$. Suppose to the contrary that $D_2\cup D_5\neq \es$. By symmetry, we may assume that $D_2\neq \es$, and let $d_2\in D_2$. Note that $d_2$ is anticomplete to $\{b_1,b_3,b_5\}$, by \ref{bto}. Now if $b_1b_3\notin E(G)$, then $\{v_5,b_3,d_2,b_1,v_2\}$ induces a $3P_1+P_2$; so $b_1b_3\in E(G)$. Similarly, if $b_3b_5\notin E(G)$, then  $\{v_2,b_3,d_2,b_5,v_5\}$ induces $3P_1+P_2$; so  $b_3b_5\in E(G)$. But then $\{b_3,v_4,v_5,v_1,b_1,b_5,d_2\}$ induces an $F_1$ which is a contradiction to our assumption that $G$ is $F_1$-free. So \cref{F3-D125} holds. $\sq$

Next:
\begin{claim}\label{F3-D34}
  We may assume that  $D_3\cup D_4$ is an empty set.
 \end{claim}
\no{\it Proof of  \cref{F3-D34}}.~Suppose that $D_3\neq \es$. We let $D_3=\{d_3\}$, by \ref{df}.  Now if there are two vertices in $N(d_3)\cap B_5$, say $b_5$ and $b_5'$, then we see that $\{v_5,b_5,b_5',d_3\}$ induces a $K_4-e$ (by \ref{X-matching}); so $|N(d_3)\cap B_5|\leq 1$. Then from \ref{DZ}, \ref{bto},  \ref{df} and \cref{F3-D125}, we observe that $deg(d_3) =|\{v_2,v_4\}|+|N(d_3)\cap B_5|+|N(d_3)\cap D_4|+|N(d_3)\cap Z_5| \leq 2+1+1+1 =5$,  and we are done. So we may assume that  $D_3 =\es$. Likewise, we may assume that  $D_4=\es$.   $\sq$

Next:
\begin{claim}\label{F3-Z25}
  We may assume that  $Z_2\cup Z_4$ is an empty set.
 \end{claim}
\no{\it Proof of  \cref{F3-Z25}}.~Suppose that $Z_2\neq \es$. We let $Z_2=\{z_2\}$, by \ref{DZ}. Note that $z_2$ is anticomplete to $B_1\cup B_3\cup B_5$, by \ref{bto}. Now if there is a vertex, say $x~(\neq b_5)$ in $X_5$, then $xb_5\in E(G)$ (by \ref{X-matching}), and then $\{z_2,b_1,v_4,b_5,x\}$ induces a $3P_1+P_2$ (by \ref{bto}); so $X_5=B_5=\{b_5\}$ and hence $|X_5|=1$. A similar argument shows that $B_1=\{b_1\}$ and hence $|B_1|= 1$.
Hence from the above claims  and from \ref{DZ},  we have $deg(v_1)=|\{v_2, v_5\} \cup B_1 \cup X_5 \cup Z_1 \cup Z_3| = |\{v_2, v_5\}|+|B_1|+|X_5|+|Z_1|+|Z_3|\leq 2+1+1+1+1 = 6$, and we are done. So we may assume that  $Z_2 =\es$. Likewise, we may assume that  $Z_4=\es$.   $\sq$

Next:
  \begin{claim}\label{F3-B24}
  We may assume that  $B_2\cup B_4$ is an empty set.
 \end{claim}
\no{\it Proof of  \cref{F3-B24}}.~Suppose that $B_2\neq \es$. Let $b_2 \in B_2$.  Then $b_2$ is anticomplete to $\{b_1,b_3\}$, by \ref{bto}.  Now if $b_1b_3\in E(G)$, then $\{v_1, b_1, b_3, v_4, v_5, b_2\}$ induces  a $C_5+K_1$; so   $b_1b_3\notin E(G)$. Also if $b_2b_5\in E(G)$, then  $\{b_2,v_3,v_4,v_5,b_5,b_1\}$ induces  a $C_5+K_1$; so  $b_2b_5\notin E(G)$.
Now if there is a vertex in $B_2\sm \{b_2\}$, say $b_2'$, then $b_2b_2'\in E(G)$ (by \ref{X-matching}), and then $\{b_1,b_5,v_4,b_2,b_2'\}$ induces a $3P_1+P_2$ (by \ref{bto}); so   $B_2=\{b_2\}$ and hence $|B_2|=1$. A similar argument shows that $B_1=\{b_1\}$ and hence  $|B_1|=1$. Thus from the above claims and from \ref{DZ},  we have $deg(v_2)=|\{v_1, v_3\}|+ |B_1|+|B_2|+|Z_1|\leq 2+1+1+1  = 5$,  and we are done. So we may assume that  $B_2 =\es$. Likewise, we may assume that  $B_4=\es$.   $\sq$

\medskip
 From the above conclusions, we observe that $V(G) = C \cup X_1 \cup X_3 \cup X_5$.    Also recall that   $[X_1, X_5]=\es$ (by \ref{bto}), and $\{X_1 , X_3\}$ and $\{X_3, X_5\}$ are graded (by \ref{X-matching}). Now if $|X_1|\leq 2$, then $deg(v_2)= |N(v_2)| = |\{v_1,v_3\} \cup X_1| \leq 4$, and we are done; so we may assume that $|X_1|\geq 3$. Similarly, we may assume that $|X_5|\geq 3$.  Moreover, if $|X_3|\leq 2$, then $deg(v_3)= |N(v_3)| = |\{v_2,v_4\} \cup X_3\cup Z_5| \leq 5$ (by \ref{DZ}), and we are done; so we may assume that $|X_3|\geq 3$.    Hence we conclude that $G \in {\cal C}$ by taking $u_i:=v_i$ for $i\in [5]$, and $Q_1:=B_1$, $Q_1':=Z_1$, $Q_2:=B_3$, $Q_2':=Z_3$, $Q_3:=B_5$ and $Q_3':=Z_5$, and we are done. This proves \cref{F3:thm}.
 \end{proof}

\begin{lemma}\label{F4:thm}
    Let $G$ be a  ($3P_1+P_2, K_4-e$)-free graph. If $G$ contains  an $F_4$, then  $G$ has a vertex of degree  at most 6 or $G$ is a good graph with $\omega(G)\geq 4$.
\end{lemma}
\begin{proof}
Let $G$ be a  ($3P_1+P_2, K_4-e$)-free graph.  From \cref{NC5:deg6,F1:thm}, we may assume that $G$ is ($F_1,   C_5+K_1$)-free. Suppose that $G$ contains  an $F_4$  as shown in Figure~\ref{fig-F1234}. We let  $C:= \{v_1, v_2, v_3, v_4, v_5\}$. Then, with respect to $C$, we define the sets $A$, $B$, $D$, $Z$, $X$ and $T$ as in Section~\ref{genprop}, and we use the properties in Section~\ref{genprop} and in \cref{3P1P2-prop}.   Clearly  $b_1\in B_1$, $b_3\in B_3$ and $a_5\in A_5$. Since $G$ is $F_1$-free, we have $A\sm A_5=\es$ and $X_4\cup X_5=\es$, and since $G$ is ($C_5+K_1$)-free, we have $T=\es$.  Moreover:
\begin{claim} \label{F4-A5}
We may assume that $|A_5|\geq 2$.
\end{claim}
\no{\it Proof of  \cref{F4-A5}}.~If $|A_5|\leq 1$, then $deg(v_5) = |N(v_5)| = |\{v_1, v_4\}| + |A_5|+|D_1| + |D_4| +|Z_2|  \leq 2 + 1+1 + 1 +1 = 6$ (by \ref{DZ} and \ref{df}), and we are done; so we may assume that $|A_5|\geq 2$. $\sq$

Next:
\begin{claim}\label{F4-B13}
We may assume that $|B_1|\geq 2$ and $|B_3|\geq 2$. Hence, there are vertices, say $b_1', b_1'' \in B_1$ and $b_3', b_3''\in B_3$ such that $b_1'b_1'', b_3'b_3''\in E(G)$ and $b_1'b_3'\not\in E(G)$.
\end{claim}
\no{\it Proof of  \cref{F4-B13}}.~If $|B_1|= 1$, then by using \ref{DZ} and \ref{df} and  since $B_3\neq \es$, we have $|D_2 \cup Z_3|\leq 1$ (by \ref{zzd}), and then $deg(v_1) = |N(v_1)| = |\{v_2, v_5\}| + |B_1| + |D_2 \cup Z_3|+|Z_1| + |D_5| \leq 2 + 1 + 1+1 + 1 = 6$, and we are done; so we may assume that $|B_1|\geq 2$. Likewise, we may assume that $|B_3|\geq 2$. Hence there exist vertices $b_1', b_1'' \in B_1$ and $b_3', b_3''\in B_3$ such that $b_1'b_1'', b_3'b_3''\in E(G)$ and $b_1'b_3'\not\in E(G)$, by \ref{X-matching}. $\sq$

Next:
\begin{claim}\label{F4-B2}
We may assume that $B_2$ is an empty set.
\end{claim}
\no{\it Proof of  \cref{F4-B2}}.~If there is a vertex in $B_2$, say $b_2$, then from \cref{F4-B13} and \ref{bto}, we see that $\{b_1'',v_1,v_5,v_4,b_3', b_2\}$ induces a $C_5+K_1$ (when $b_1''b_3'\in E(G)$) or   $\{b_2,b_3',v_5,b_1',b_1''\}$  induces a $3P_1+P_2$  (when $b_1''b_3'\notin E(G)$); so we may assume that $B_2=\es$. $\diamond$

Next:
\begin{claim}\label{F4-D235}
We may assume that  $D_2\cup D_3\cup D_5$ is an empty set.
\end{claim}
\no{\it Proof of  \cref{F4-D235}}.~If there is a vertex in $D_2\cup D_3\cup D_5$, say $d$, then from \ref{F4-B13} and \ref{bto}, we observe that  $\{d, v_5, b_3', b_1', v_2\}$ or $\{d, v_5, b_1',v_3, b_3'\}$  induces a $3P_1+P_2$; so we may assume that  $D_2\cup D_3\cup D_5 =\es$. $\sq$

\medskip
Further, we have the following.
\begin{claim}\label{F4-A5D1D4}
The set $A_5\cup D_1\cup D_4\cup \{v_5\}$ is a clique.
\end{claim}
\no{\it Proof of  \cref{F4-A5D1D4}}.~If there are nonadjacent vertices, say $d_1\in D_1$ and $d_4\in D_4$, then $\{d_1,a_5,v_5,d_4\}$  induces a $K_4-e$ (by \ref{ato}); so  $[D_1, D_4]$ is complete, and hence $ A_5\cup D_1\cup D_4\cup \{v_5\}$ is a clique, by \ref{df}, \ref{abf} and \ref{ato}.  $\sq$

\begin{claim}\label{F4-A5Z5}
The set $ A_5\cup Z_2\cup \{v_5\}$ is a clique.
\end{claim}
\no{\it Proof of  \cref{F4-A5Z5}}.~If there are nonadjacent vertices, say $a\in A_5$ and $z\in Z_2$, then pick a vertex, say $a'\in A_5\sm \{a\}$ (such a vertex $a'$ exists (by \cref{F4-A5}) and recall that $aa'\in E(G)$, by \ref{abf}), and then   we see that $\{v_1,z,v_4, a,a'\}$ induces a $3P_1+P_2$  or $\{z,a,a',v_5\}$ induces a $K_4-e$; so $[A_5,Z_2]$ is complete. Hence from \ref{DZ} and \ref{abf}, we conclude that $A_5\cup Z_2\cup \{v_5\}$ is a clique. $\sq$

\medskip
To proceed further, we let $Q_1:=  A_5 \cup D_1 \cup D_4\cup Z_2\cup \{v_5\}$, $Q_2:= \{v_1, v_2\} \cup X_1$ and $Q_3:= \{v_3, v_4\} \cup X_3$, and we claim the following.

\begin{claim}\label{F4graded}
The sets $Q_1$, $Q_2$ and $Q_3$ are pairwise graded.
\end{claim}
\no{\it Proof of  \cref{F4graded}}.~By \ref{X-matching} and by symmetry, it is enough to prove that  $\{Q_1, Q_2\}$ is graded. Recall that since $A_5\neq \es$, one of $D_1\cup D_4$ and $Z_2$ is empty, by \ref{afd}. So from \cref{F4-A5D1D4} and  \cref{F4-A5Z5}, we see that $Q_1$ is a clique.  Thus, from \ref{X-matching}, it is enough to prove that
    $[Q_1,  Q_2]$ is special. (We note that the proof is   similar to the proof of  \cref{A1graded} of \cref{lem-2k1p3-c5}.)

        First we show that each vertex in $Q_1$ is adjacent to at most one vertex in $Q_2$. Suppose to the contrary that there is a vertex, say $q_1\in Q_1$, such that $|N(q_1)\cap Q_2|\geq 2$. Let $q_2,q_2'\in N(q_1)\cap Q_2$. Then clearly $q_1\notin \{v_5\}\cup D_1\cup Z_2$, by \ref{bto}; so $q_1\in A_5\cup D_4$. Then since $A_5\cup D_4$ is anticomplete to $\{v_1,v_2\}$, clearly $\{q_2,q_2'\}\subseteq X_1$. But then $\{q_1,q_2,q_2',v_2\}$ induces a $K_4-e$ which is a contradiction. So each vertex in $Q_1$ is adjacent to at most one vertex in $Q_2$.

    Next we show that each vertex in $Q_2$ is adjacent to at most one vertex in $Q_1$. Suppose to the contrary that there is a vertex, say $q_2\in Q_2$, such that $|N(q_2)\cap Q_1|\geq 2$. Let $q_1,q_1'\in N(q_2)\cap Q_1$. Clearly $q_2\neq v_1$, and  since $|N(v_2)\cap Q_1| \leq |D_1\cup Z_2|\leq 1$ (by \ref{df}, \ref{DZ} and \ref{afd}), we have $q_2\neq v_2$. So $q_2\in X_1$. Then since $Q_1\sm \{v_5\}$ is complete to $v_5$, we see that $\{q_2,q_1,q_1',v_5\}$ induces   a $K_4-e$ which is a contradiction. So each vertex in $Q_2$ is adjacent to at most one vertex in $Q_1$. This proves \cref{F4graded}. $\sq$

\medskip
Thus   from the above arguments, we see that $V(G) = C \cup A_5 \cup X_1 \cup X_3 \cup D_1 \cup D_4   \cup Z_2 =Q_1\cup Q_2\cup Q_3$.  Also since $|B_1|\geq 2$ (by \cref{F4-B13}), and since $B_1\cup \{v_1,v_2\}$ is a clique (by \ref{X-matching}), clearly  $\omega(G) \geq 4$.     Then from \ref{F4graded}, we see that $Q_1, Q_2$ and $Q_3$ are three cliques which are pairwise graded  whose union is $V(G)$.   Hence $G$ is a good graph with $\omega(G)\geq 4$. This completes the proof.
\end{proof}

\begin{lemma}\label{C5:thm}
Let $G$ be a  ($3P_1+P_2, K_4-e$)-free graph. If $G$ contains  a $C_5$, then    $G$ has a vertex of degree  at most 6 or $G$ is a good graph with $\omega(G)\geq 4$  or $G\in \cal C$.
\end{lemma}
\begin{proof}
Let $G$ be a  ($3P_1+P_2, K_4-e$)-free graph.  From \cref{NC5:deg6,F1:thm,F2:thm,F3:thm,F4:thm}, we may assume that $G$ is ($F_1, F_2, F_3, F_4, C_5+K_1$)-free.  Suppose that $G$ contains  a  $C_5$, say with vertex-set $C:= \{v_1, v_2, v_3, v_4, v_5\}$ and edge-set  $\{v_1 v_2, v_2 v_3, v_3 v_4, v_4 v_5, $ $v_5 v_1\}$.   Then, with respect to $C$, we define the sets $A$, $B$, $D$, $Z$, $X$ and $T$ as in Section~\ref{genprop}, and we use the properties in Section~\ref{genprop}  and in \cref{3P1P2-prop}.
Moreover:

 \begin{claim}\label{C5-B13}
  We may assume that  for each $i\in [5]$,  one of $B_i$ and $B_{i+2}$ is empty.
 \end{claim}
   \no{\it Proof of \cref{C5-B13}}.~Suppose that there  exists an index $i\in [5]$  such that $B_i$ and $B_{i+2}$ are nonempty. By symmetry, we may assume that $i=1$.   Then since $G$ is $F_3$-free, we have $B_4\cup B_5=\es$. Also  since $G$ is $F_4$-free, we have $A_5=\es$. Now since $B_1 \neq \es$, one of $D_1$ and $Z_2$ is empty (by  \ref{zzd}), and hence $|D_1\cup Z_2| \leq 1$, by \ref{DZ} and \ref{df}. So $deg(v_5) = |N(v_5)| = |\{v_1, v_4\}| + |D_1 \cup Z_2| + |D_4| + |Z_4| + |Z_5| =2+1+1+1+1 \leq 6$ (by \ref{DZ} and \ref{df}), and we are done. $\sq$

Next:
\begin{claim}\label{C5-B15}
  We may assume that  for each $i\in [5]$, one of $B_i$ and $B_{i+1}$ is empty.
 \end{claim}
 \no{\it Proof of \cref{C5-B15}}.~Suppose that there  is an index $i\in [5]$  such that $B_i$ and $B_{i+1}$ are nonempty. By symmetry, we may assume that  $i=5$. Let $b_1\in B_1$ and $b_5\in B_5$.  Then by using \cref{C5-B13}, we see that $B_2\cup B_3\cup B_4 = \es$. Also since $G$ is $F_1$-free, we have $A_1\cup A_2\cup A_5=\es$. Recall that $[B_1, B_5]=\es$, by \ref{bto}. Moreover we have the following.

\begin{subclaim}
 \label{C5-B15-D2} The set $D_2$ is an empty set.
 \end{subclaim}
 \no{\em Proof of \cref{C5-B15-D2}}.~If there is a vertex in $D_2$, say $d_2$,  then   $\{v_1, d_2, v_3, v_4, v_5,b_5,b_1\}$   induces an $F_1$ (by \ref{bto}); so we conclude that $D_2=\es$. $\diamond$

 \begin{subclaim}
 \label{C5-A3}
  We have $|A_3|\leq 1$.
  \end{subclaim}
  \no{\em Proof of \cref{C5-A3}}.~First we show that $[A_3, B_1]$ is  complete. If there are vertices, say $a_3\in A_3$ and $b_1'\in B_1$, such that $a_3b_1'\notin E(G)$, then $\{b_5,a_3,v_4,b_1',v_2\}$ induces a $3P_1+P_2$ (when $a_3b_5\notin E(G)$) or $\{b_5,a_3,v_3,v_4,v_5,b_1'\}$ induces a $C_5+K_1$ (when $a_3b_5\in E(G)$); so we conclude that $[A_3, B_1]$ is  complete. Hence if there are two vertices in $A_3$, say $a_3$ and $a_3'$,  then $\{b_1,a_3,a_3',v_3\}$ induces a $K_4-e$ (by \ref{abf}); so we may assume that $|A_3|\leq 1$. $\diamond$

  \medskip
   Now from the above conclusions and since $|D_4 \cup Z_3 \cup Z_5| \leq 2$ (by  \ref{zzd}), we see that $deg(v_3) = |\{v_2, v_4\}| + |D_4 \cup Z_3 \cup Z_5| + |A_3|+|Z_2| \leq 2 + 2 + 1+1 = 6$, by \ref{DZ}, and  we are done. $\sq$

Next:
   \begin{claim}\label{C5-B}
 We may assume that $B$ is an empty set.
   \end{claim}
 \no{\it Proof of \cref{C5-B}}.~Suppose that $B\neq \es$. Then from \cref{C5-B13} and \cref{C5-B15}, we may assume that there is an index $i\in [5]$  such that $B_i\neq \es$ and $B\sm B_i=\es$. By symmetry, we may assume that $i=5$.    Then since $G$ is $F_1$-free, we have $A_1\cup A_5 =\es$, and so $A = A_2 \cup A_3 \cup A_4$. Also since $G$ is $F_2$-free,  there exists an index $j\in \{2,3,4\}$ and $j$ mod 5  such that $A_j=\es$, and then $deg(v_j) = |\{v_{j-1}, v_{j+1}\}|+|D_{j-1}|+|D_{j+1} \cup Z_j \cup Z_{j+2}|+|Z_{j-1}| \leq 2+1+2+1 =6$ (by \ref{DZ}, \ref{df} and \ref{zzd}), and we are done. $\sq$

\medskip
Now since $G$ is $F_2$-free, there exists an index $i\in [5]$ such that $A_i=\es$. By symmetry, we may assume that $i=2$. Recall that  $|D_3 \cup Z_2 \cup Z_4| \leq 2$, by \ref{zzd}. So from \cref{C5-B}, we conclude that $deg(v_2) = |\{v_1, v_3\}| + |D_1| + |D_3 \cup Z_2 \cup Z_4| + |Z_1| \leq 2 + 1 + 2 + 1 = 6$ , by \ref{DZ} and \ref{df}. This completes the proof of \cref{C5:thm}.
 \end{proof}

 \begin{lemma}\label{C7:thm}
Let $G$ be a  ($3P_1+P_2, K_4-e, C_5$)-free graph. If $G$ contains  a $C_7$, then    $G$ has a vertex of degree  at most 6.
\end{lemma}
\begin{proof}
Let $G$ be a ($3P_1+P_2, K_4-e, C_5$)-free graph. Suppose that $G$ contains an induced $C_7$, say with
 vertex-set $C:= \{v_1, v_2, ,\ldots, v_6, v_7\}$ and edge-set $\{v_1 v_2, v_2 v_3, \ldots,  v_6 v_7, v_7v_1\}$.  For $i\in [7]$ and $i$ mod $7$, we let  $Y_i:=\{y\in V(G)\sm C \mid yv_i, yv_{i+1}\in E(G)$ and $yv_{i+2}, yv_{i-1}, yv_{i-3}\notin E(G)\}$, and let $Y:=\cup_{i=1}^7Y_i$. Then we have the following.
 \begin{claim}\label{C7-nei}
 $V(G) = C\cup Y$.
 \end{claim}
 \no{\it Proof of \cref{C7-nei}}.~Suppose to the contrary that  $V(G)\sm (C\cup Y)\neq \es$, and let $u\in  V(G)\sm (C\cup Y)$.  First we show that:
    \begin{equation}\label{3P1P2-C7-eq}
  \longbox{\em We may assume that for each $i\in [7]$ and $i$ mod $7$,   $u$ is nonadjacent to at least one of $v_i$ and $v_{i+1}$.} \tag*{($\star \star$)}
   \end{equation}
 \no{\it Proof of} \ref{3P1P2-C7-eq}.~Suppose to the contrary that there is an index $i\in [7]$  and $i$ mod $7$ such that $uv_i,uv_{i+1}\in E(G)$.
Then $uv_{i+2}\notin E(G)$, for otherwise,   $\{u,v_i,v_{i+1},v_{i+2}\}$ induces $K_4-e$. Similarly, we have  $uv_{i-1}\notin E(G)$.
Now if $uv_{i-3} \in E(G)$, then since $\{v_i,u,v_{i-3},v_{i-2},v_{i-1}\}$ and $\{v_{i+1},u,v_{i-3},v_{i+3},v_{i+2}\}$ do not induce $C_5$'s, we have $uv_{i-2},uv_{i+3}\in E(G)$, and then $\{v_{i-2},u,v_{i-3},$ $v_{i+3}\}$ induces a $K_4-e$; so $uv_{i-3} \notin E(G)$, and hence $u\in Y_i$ which is a contradiction. So we may assume that for each $i\in [7]$,   $u$ is nonadjacent to at least one of $v_i$ and $v_{i+1}$. This proves  \ref{3P1P2-C7-eq}. $\diamond$

\medskip
 Now since $G$ is  ($3P_1+P_2$)-free, clearly $u$ has neighbor in $C$, say $uv_1\in E(G)$. Then by  \ref{3P1P2-C7-eq}, we see that $uv_2,uv_7\notin E(G)$. Thus, if $u$ is  anticomplete to $\{v_4,v_5\}$, then  $\{v_7,u,v_2,v_4,v_5\}$ induces a $3P_1+P_2$; so $u$ is adjacent to one of $v_4$ and $v_5$.  By symmetry, we may assume that  $uv_4\in E(G)$. But then from \ref{3P1P2-C7-eq}, we have $uv_3\notin E(G)$,  and then we see that $\{v_1,v_2,v_3,v_4,u\}$ induces a $C_5$  which is a contradiction. So $V(G)\sm (C\cup Y)=\es$, and hence \cref{C7-nei} holds. $\sq$

\begin{claim}\label{C7-Yicard}
For each $i\in [7]$ and $i$ mod $7$, we have $|Y_i|\leq 1$.
 \end{claim}
 \no{\it Proof of \cref{C7-Yicard}}.~If there are  two vertices in $Y_i$, say $y$ and $y'$, then  $\{y,v_i,v_{i+1},y'\}$  induces a $K_4-e$ (when $yy'\notin E(G)$) or  $\{v_{i+2}, v_{i-1}, v_{i-3},y,y'\}$ induces a $3P_1+P_2$ (when $yy'\in E(G)$); so  $|Y_i|\leq 1$. $\sq$

\medskip
Now from the above conclusions, we observe that $deg(v_1)\leq |\{v_2,v_7\}|+|Y_1|+|Y_3|+|Y_5|+|Y_7|\leq 6$. This completes the proof.
 \end{proof}

\medskip
\noindent{\bf Proof of \cref{main-thm}}.
Let $G$ be a  ($3P_1+P_2, K_4-e$)-free graph which is not perfect.  Then by \cref{spgt}, $G$ contains  an odd-hole or the complement graph of an odd-hole. Since the complement graph of any odd-hole of length at least 7 contains a $K_4-e$,  and since any odd-hole of length at least 9 contains a $3P_1+P_2$, we may assume that $G$ contains   a $C_5$ or a $C_7$. If $G$ contains a $C_5$, then the theorem follows from \cref{C5:thm}, and if $G$ is $C_5$-free and contains a $C_7$, then the theorem follows from \cref{C7:thm}. This proves \cref{main-thm}.
\hfill{$\Box$}

\subsection{Coloring of ($3P_1+P_2$,\,$K_4-e$)-free graphs}\label{Sec:col}

In this subsection, we prove \cref{col-thm} and its consequences.

\medskip
\noindent{\bf Proof of \cref{col-thm}}.~ Let $G$ be a  ($3P_1+P_2, K_4-e$)-free graph. First we prove \cref{col-thm}:$(i)$. Since $\omega(G)=3$, we may assume that $G$ is $K_4$-free.  Let $v$ be any vertex in $G$.
 Then   clearly $G[N(v)]$ is a ($K_3, P_3$)-free graph which is a  bipartite graph, and   $G[\overline{N}(v)]$ is  a ($2P_1+P_2, K_4-e, K_4$)-free graph.
It follows from a result of Gy\'arf\'as \cite{Gyarfas} that  $\chi(G[\overline{N}(v)])\leq 3$. So $\chi(G)\leq \chi(G[N(v)])+\chi(G[\overline{N}(v)\cup\{v\}])\leq 2+3 =5$, and we are done.

 Next we prove \cref{col-thm}:$(ii)$. Let $G$ be a  ($3P_1+P_2, K_4-e$)-free graph with $\omega(G)\geq 4$, and let $\ell := \max\{7, \omega(G)\}$. We prove \cref{col-thm}:$(ii)$ by induction on $|V(G)|$. From  \cref{main-thm}, \cref{lem-good} and \cref{classC-col}, we may assume that $G$ has a vertex of degree  at most 6, say $v$. By induction hypothesis, we have $\chi(G-\{v\})\leq  \ell$. Since $v$ has at most 6 neighbors in $G-\{v\}$,  we can take any $ \ell$-coloring
of $G-\{v\}$ and extend it to a $\ell$-coloring of $G$, using for $v$ a (possibly new) color
that does not appear in its neighborhood. This proves \cref{col-thm}:$(ii)$.

Clearly the bound in \cref{col-thm} is  tight when $\omega(G)\geq 7$ (for instance, take $G\cong K_t$ where $t\geq 7$).
\hfill{$\Box$}

\medskip
We remark that   the bound given in \cref{col-thm} does not seem to be tight when $\omega(G)\in \{3,4,5,6\}$. Note that there are ($3P_1+P_2, K_4-e$)-free graphs with chromatic number equal to $4$ and clique number equal to 3; see \cref{fig-opt} for such graphs.

\begin{cor}\label{3P1P2-NOC}
The class of ($3P_1+P_2, K_4-e$)-free graphs is near optimal colorable. That is, every ($3P_1+P_2, K_4-e$)-free graph $G$ satisfies $\chi(G)\leq \max\{7, \omega(G)\}$.
\end{cor}
\begin{proof}
Since every ($3P_1+P_3, K_3$)-free graph $G$ satisfies $\chi(G)\leq 3$, by a result of Randerath \cite{Rand-Thesis}, the proof follows from \cref{col-thm}.
\end{proof}

  Dabrowski,  Golovach and  Paulusma \cite{DGP-TCS14} showed that \textsc{Coloring} for  the class of ($3P_1+P_2$, $K_4-e$)-free graphs is solvable in  polynomial time   for every fixed positive integer $k \leq 6$. This result together with \cref{3P1P2-NOC} and \cref{JH-thm}  imply  the following.

 \begin{cor}
  \textsc{Chromatic Number} for
	the class of ($3P_1+P_2$, $K_4-e$)-free graphs can be solved in  polynomial time.
 \end{cor}

\medskip
\noindent{\bf Acknowledgement}.  The authors would like to thank Prof.~Daniel Paulusma and Prof.~Maria Chudnovsky for   fruitful  discussions.
The authors thank the anonymous referees for their valuable suggestions and comments which improved the presentation of the paper.
The first author acknowledges the National Board of Higher Mathematics (NBHM), Department of Atomic Energy, India for the financial support to carry out this research work, and CHRIST (Deemed to be University), Bengaluru for permitting to pursue the NBHM project.


{\small
\begin{thebibliography}{99}
\bibitem{BCMSSZ}F. Bonomo, M. Chudnovsky, P. Maceli, O. Schaudt, M. Stein and M. Zhong, Three-coloring and list three-coloring of graphs without induced
paths on seven vertices. Combinatorica 38(4) (2017) 1--23.

\bibitem{BDHP}A. Brandst\"adt, K. K. Dabrowski, S. Huang, D. Paulusma, Bounding the clique-width of $H$-free chordal graphs. Journal of Graph Theory 86 (2017)  42--77.

\bibitem{BLS}A. Brandst\"adt, V. B. Le and J. P. Spinrad, Graph classes: A survey. SIAM monographs on Discrete Mathematics
and Applications, Philadelphia (2004).

\bibitem{BRSV} C.~Brause, B. Randerath, I. Schiermeyer and E. Vumar, On the chromatic number
of $2K_2$-free graphs. Discrete Applied Mathematics  253 (2019) 14--24.\

\bibitem{AK-Survey}A. Char and T. Karthick, $\chi$-boundedness and related problems on graphs without long induced paths: A survey.
Discrete Applied Mathematics  364 (2025) 99--119.

\bibitem{AK}A. Char and T. Karthick, On graphs with no induced $P_5$ or $K_5-e$.
 Journal of Graph Theory  110:1 (2025) 5--22.

\bibitem{CKLV}M. Chudnovsky, A. Kabela, B. Li and P. Vr\'ana, Forbidden induced pairs for perfectness and
$\omega$-colourability of graphs. The Electronic Journal of Combinatorics 29(2) (2022) \#P2.21.

\bibitem{spgt}M. Chudnovsky, N. Robertson, P. Seymour and R. Thomas, The strong perfect graph theorem. Annals of Mathematics 164  (2006) 51--229.

\bibitem{DDP}K. K. Dabrowski,  F. Dross and D. Paulusma, Colouring diamond-free graphs. Journal of Computer and System Sciences 89 (2017) 410--431.

\bibitem{DGP-TCS14} K. K. Dabrowski, P. A. Golovach and D. Paulusma, Colouring of graphs with Ramsey-type forbidden subgraphs. Theoretical Computer Science 522 (2014) 34--43.

\bibitem{DHP}K. K. Dabrowski, S. Huang and D. Paulusma, Bounding clique-width via perfect graphs. Journal of Computer and System Sciences 104 (2019) 202--215.



   \bibitem{Geiber} M. Gei\ss er, Colourings of $P_5$-free graphs. PhD. Thesis. Technische Universit\"at Bergakademie
Freiberg (2022).

    \bibitem{GHJM}J. Goedgebeur, S. Huang, Y. Ju  and O. Merkel, Colouring graphs with no induced six-vertex
path or diamond. Theoretical Computer Science  941 (2023) 278--299.

\bibitem{GJPS}P.~A.~Golovach, M.~Johnson, D.~Paulusma and J.~Song, A survey on the computational complexity of colouring graphs with forbidden subgraphs. Journal of Graph Theory 84 (2017) 331--363.

    \bibitem{GLS} M. Gr\"otschel, L. Lov\'asz  and A. Schrijver, The ellipsoid method and its consequences in
combinatorial optimization. Combinatorica 1 (1981) 169-197. Corrigendum: Combinatorica 4 (1984) 291--295.

		\bibitem{Gyarfas}A. Gy\'arf\'as, Problems from the world surrounding perfect graphs. Zastosowania Matematyki Applicationes Mathematicae  19 (1987) 413--441.

\bibitem{Ju-Huang}Y.~Ju and S.~Huang, Near optimal colourability on hereditary graph
families. Theoretical Computer Science 993 (2024) Article No.: 114465.

\bibitem{Ju-Huang25}Y.~Ju and S.~Huang, Near optimal colourability on ($H, K_n-e$)-free graphs. Discrete Applied Mathematics
 367 (2025)  1--7.


\bibitem{KKTW}D. Kr\'al, J. Kratochv\'il, Zs. Tuza and G. J. Woeginger, Complexity of coloring graphs without forbidden
induced subgraphs. In: WG 2001, Lecture Notes in Computer Science 2204  (2001) 254--262.

\bibitem{Lovasz}L. Lov\'asz, A characterization of perfect graphs. Journal of Combinatorial Theory, Series
B 13 (1972) 95--98.


\bibitem{LM}V.V. Lozin and D.S. Malyshev, Vertex coloring of graphs with few obstructions. Discrete Applied Mathematics 216 (2017) 273--280.




\bibitem{Olariu}S. Olariu, Paw-free graphs. Information Processing Letters 28 (1988) 53–54.

\bibitem{Rao}  M. Rao, MSOL partitioning problems on graphs of bounded treewidth and clique-width. Theoretical
Computer Science 377 (2007)  260--267.

\bibitem{Rand-Thesis}B. Randerath, The Vizing bound for the chromatic number based on forbidden pairs.
Ph. D. Thesis, RWTH Aachen, Shaker Verlag 1998.

\bibitem{Rand}B. Randerath, Three-Colorability and forbidden subgraphs. I: Characterizing pairs. Discrete
Mathematics 276 (2004) 313--325.

\bibitem{RST}B. Randerath, I. Schiermeyer and M. Tewes, Three-colorability and forbidden subgraphs. II:
Polynomial algorithms. Discrete Mathematics 251 (2002) 137–153.

\bibitem{SR-Poly-Survey}
I. Schiermeyer and B. Randerath, Polynomial $\chi$-binding functions and forbidden induced subgraphs: A survey. Graphs and Combinatorics 35 (2019) 1--31.

\bibitem{Schn} D. Schindl, Some new hereditary classes where graph coloring remains NP-hard.  Discrete
Mathematics 295 (2005) 197--202.



\bibitem{west}D. B. West. Introduction to Graph Theory (2nd Edition). Prentice-Hall, Englewood
Cliffs, New Jersey, 2000.

\end{thebibliography}
}
\end{document}